\documentclass[12pt]{article}

\ifx\comments\undefined
  \usepackage[letterpaper,margin=1in,includefoot]{geometry}
  \newcommand{\XXXcomment}[1]{}
\else
  \usepackage[a4paper,margin=.5in,right=2.5in,marginpar=2in,includefoot]{geometry}
  \newcommand{\XXXcomment}[1]{\marginpar{\color{blue}{\footnotesize #1}}}
\fi

\usepackage{latexsym,amssymb,amsmath,amscd,amsthm,alltt,stmaryrd,graphicx,mathrsfs}

\usepackage{times}

\usepackage{float}
\floatstyle{boxed}
\restylefloat{figure}
\restylefloat{table}

\usepackage{tikz}
\usetikzlibrary{trees,arrows,positioning,patterns,automata,shapes,decorations,decorations.pathmorphing}

\newcommand{\mystretch}{\renewcommand{\arraystretch}{1.5}}
\newcommand{\normalstretch}{\renewcommand{\arraystretch}{1}}

\newlength{\mysep}
\newlength{\mysepCloser}
\newlength{\mysepFarther}
\newlength{\mysepFarthest}
\newlength{\myMinsize}
\newlength{\myBigger}
\newlength{\myNodeDistance}
\newenvironment{mytikz}[1][0em]{
\begin{tikzpicture}[>=latex,auto,node distance=\mysep,
  baseline={([yshift=#1]current bounding box.east)}]

  \normalstretch{}

  \tikzstyle{w}=[draw,circle,thick,minimum size=\myMinsize]

  \tikzstyle{e}=[draw,minimum size=\myMinsize,node distance=\myNodeDistance]
  
  \tikzstyle{every edge}=[draw,thick,font=\footnotesize]
  
  \tikzstyle{every label}=[font=\footnotesize]
  
  \tikzstyle{ev}=[anchor=west,node distance=\myNodeDistance]

  \tikzstyle{bigger}=[minimum size=\myBigger]

  \tikzstyle{closer}=[node distance=\mysepCloser]  

  \tikzstyle{farther}=[node distance=\mysepFarther]

  \tikzstyle{farthest}=[node distance=\mysepFarthest]

  \tikzstyle{l}=[node distance=\myNodeDistance]
}{\mystretch{}\end{tikzpicture}}

\theoremstyle{definition}
\newtheorem{theorem}{Theorem}[section]
\newtheorem{corollary}[theorem]{Corollary}

\newtheorem{definition}[theorem]{Definition}
\newtheorem{lemma}[theorem]{Lemma}

\newtheorem{example}[theorem]{Example}

\newtheorem{remark}[theorem]{Remark}

\newcommand{\Rat}{\mathbb{Q}}  
\newcommand{\Int}{\mathbb{Z}}  

\newcommand{\M}{{\cal M}}      
\newcommand{\N}{{\cal N}}      

\newcommand{\Prop}{{\bf P}}    

\newcommand{\Lang}{{\cal L}}   

\newcommand{\KB}{{\mathsf{KB}}}                 
\newcommand{\KBeq}{{\mathsf{KB.5}}}             
\newcommand{\KBeqm}{{\mathsf{KB.5}}^{-}}        

\newcommand{\MCl}{\mathsf{MCl}}                   

\newcommand{\Ceq}{{\mathcal{C}^{\mathsf{.5}}}}  

\newcommand{\sem}[1]{\llbracket{#1}\rrbracket}               

\newcommand{\modelsn}{\models_{\mathsf{n}}}                  
\newcommand{\semn}[1]{\llbracket{#1}\rrbracket_{\mathsf{n}}} 

\newcommand{\modelsp}{\models_{\mathsf{p}}}                  
\newcommand{\semp}[1]{\llbracket{#1}\rrbracket_{\mathsf{p}}} 

\newcommand{\ourtitle}{Belief as Willingness to Bet}

\newcommand{\jan}{Jan van Eijck}

\newcommand{\janAffiliation}{CWI \&\ ILLC, Amsterdam}

\newcommand{\bryan}{Bryan Renne}

\newcommand{\bryanAffiliation}{ILLC, University of Amsterdam}

\newcommand{\bryanFunding}{Funded by an Innovational Research
  Incentives Scheme Veni grant from the Netherlands Organisation for
  Scientific Research (NWO).}

\title{\ourtitle{}} 

\author{\jan{}\\{}{\small\janAffiliation{}} \and
  \bryan{}\footnote{\bryanFunding{}}\\{}{\small\bryanAffiliation{}}}

\date{}

\begin{document}

\maketitle

\begin{abstract}
  We investigate modal logics of high probability having two unary
  modal operators: an operator $K$ expressing probabilistic certainty
  and an operator $B$ expressing probability exceeding a fixed
  rational threshold $c\geq\frac 12$.  Identifying knowledge with the
  former and belief with the latter, we may think of $c$ as the
  agent's betting threshold, which leads to the motto ``belief is
  willingness to bet.''  The logic $\KBeq$ for $c=\frac 12$ has an
  $\mathsf{S5}$ $K$ modality along with a sub-normal $B$ modality that
  extends the minimal modal logic $\mathsf{EMND45}$ by way of four
  schemes relating $K$ and $B$, one of which is a complex scheme
  arising out of a theorem due to Scott.  Lenzen was the first to use
  Scott's theorem to show that a version of this logic is sound and
  complete for the probability interpretation.  We reformulate
  Lenzen's results and present them here in a modern and accessible
  form.  In addition, we introduce a new epistemic neighborhood
  semantics that will be more familiar to modern modal logicians.
  Using Scott's theorem, we provide the Lenzen-derivative properties
  that must be imposed on finite epistemic neighborhood models so as
  to guarantee the existence of a probability measure respecting the
  neighborhood function in the appropriate way for threshold
  $c=\frac 12$.  This yields a link between probabilistic and modal
  neighborhood semantics that we hope will be of use in future work on
  modal logics of qualitative probability. We leave open the question
  of which properties must be imposed on finite epistemic neighborhood
  models so as to guarantee existence of an appropriate probability
  measure for thresholds $c\neq\frac 12$.
\end{abstract}

\section{Introduction}

De Finetti \cite{deFinetti51,deFinetti37} proposed the following
axiomatization of qualitative probabilistic comparison (presented here
based on \cite{Sco64:JMP}): for sets $X$, $Y$, and $Z$ coming from the
powerset $\wp(W)$ of a nonempty finite set $W$, we have
\begin{enumerate}
\item $W\npreceq\emptyset$,

\item $\emptyset\preceq X$,
  
\item $X \preceq Y$ or $Y \preceq X$,

\item $X \preceq Y\preceq Z$ implies $X \preceq Z$, and
  
\item $X \preceq Y$ if and only if $X \cup Z \preceq Y \cup Z$ for $Z$
  disjoint from $X$ and $Y$.
\end{enumerate}
De Finetti conjectured that any binary relation $\preceq$ on $\wp(W)$
that satisfies these conditions is \emph{realizable by a probability
  measure $P$} on $\wp(W)$, which means that we have $X\preceq Y$ if
and only if $P(X)\leq P(Y)$.  While every probability measure
realizing a binary relation $\preceq$ on $\wp(W)$ satisfies de
Finetti's conditions, these conditions do not in general guarantee the
existence of a realizing probability measure: it was shown by Kraft,
Pratt, and Seidenberg \cite{KraPraSei59:AMS} (presented here as in
\cite{Segerberg1971:qpiams}) that for $W=\{a,b,c,d,e\}$, the relations
\begin{eqnarray*}
\{c\}\prec\{a,b\},
&
\{b,d\}\prec\{a,c\},
\\
\{a,e\}\prec\{b,c\},
&
\{a,b,c\}\prec\{d,e\}
\end{eqnarray*}
may be extended to a binary relation over $\wp(W)$ that satisfies de
Finetti's conditions and yet has no realizing probability measure.
Kraft, Pratt, and Seidenberg (``KPS'') also determined what was
missing from de Finetti's axiomatization; Scott \cite{Sco64:JMP} later
presented the KPS conditions in a linear algebraic form.

\begin{theorem}[{\cite[Theorem 4.1]{Sco64:JMP}} reformulated with
  $\wp(W)$ instead of a general Boolean algebra]
  Let $W$ be a nonempty finite set. Given $X\in\wp(W)$, let
  $\iota:W\to\{0,1\}$ be the characteristic function of $X$ (i.e.,
  $\iota(X)(w)=1$ if $w\in X$, and $\iota(X)(w)=0$ if $w\notin X$).
  Construe functions $x:W\to\mathbb{R}$ as vectors: $x(w)$ indicates
  the real-number value of vector $x$ at coordinate $w$. Addition and
  negation of these vectors is taken component-wise:
  $(x+y)(w):=x(w)+y(w)$ and $(-x)(w):=-(x(w))$.  A binary relation
  $\preceq$ on $\wp(W)$ is realizable by a probability measure if and
  only if it satisfies each of the following: for each
  $m\in\mathbb{Z}^+$ and $X,Y,X_1,\dots,X_m,Y_1,\dots,Y_m\in\wp(W)$,
  we have
  \begin{enumerate}
  \item $\emptyset\prec W$;

  \item $\emptyset\preceq X$;

  \item $X\preceq Y$ or $Y\preceq X$; and

  \item if $X_i\preceq Y_i$ for each $i\leq m$ and
    $\sum_{i=1}^m \iota(X_i)=\sum_{j=1}^m\iota(Y_j)$, then
    $Y_j\preceq X_j$ for each $j\leq m$.
  \end{enumerate}
\end{theorem}

Scott's fourth condition is the most difficult. The algebraic
component
\begin{equation}
  \textstyle
  \sum_{i=1}^m \iota(X_i)=\sum_{j=1}^m\iota(Y_j)
  \label{eq:Scott4th}
\end{equation}
of this condition says: for each coordinate $w\in W$, the number of
$X_i$'s that contain $w$ is equal to the number of $Y_j$'s that
contain $w$.  Intuitively, Scott's forth condition tells us that if
two length-$m$ sequences of coordinate sets are related component-wise
by the relation ``is no more probable than'' and the occurrence
multiplicity of any given world is the same in each of the sequences,
then the sets are also related component-wise by the relation ``has
the same probability.''

Using Scott's theorem to prove completeness, Segerberg
\cite{Segerberg1971:qpiams} studied a modal logic of qualitative
probability.  Segerberg's logic has a binary operator $\preceq$
expressing qualitative probabilistic comparison and a unary operator
$\Box$ expressing necessity.  G{\"a}rdenfors \cite{Gardenfors75}
considered a simplified version of Segerberg's logic that, among other
differences, eliminated the necessity operator in lieu of the
abbreviation $\Box\varphi:=(1\preceq\varphi)$, which has the semantic
meaning that $\varphi$ has probability $1$ (and implies that $\varphi$
is true at all outcomes with nonzero probability).  Both
G{\"a}rdenfors and Segerberg express the algebraic component
\eqref{eq:Scott4th} of Scott's fourth condition using Segerberg's
notation
\[
(\varphi_1,\dots,\varphi_m\mathbb{E}\psi_1,\dots,\psi_m)\enspace,
\]
which we sometimes shorten to
$(\varphi_i\mathbb{E}\psi_i)_{i=1}^m$. This expression abbreviates the
formula
\[
\Box(F_0\lor\cdots\lor F_m)\enspace,
\]
where each $F_i$ is the disjunction of all conjunctions
\[
d_1\varphi_1\land\cdots\land d_m\varphi_m\land
e_1\psi_1\land\cdots\land e_m\psi_m\land
\]
satisfying the property that exactly $i$ of the $d_k$'s are the empty
string, exactly $i$ of the $e_k$'s are the empty string, and the rest
of the $d_k$'s and $e_k$'s are the negation sign $\lnot$.
Intuitively, $F_i$ says that $i$ of the $\varphi_k$'s are true and $i$
of the $\psi_k$'s are true; $F_0\lor\cdots\lor F_m$ says that the
number of true $\varphi_k$'s is the same as the number of true
$\psi_k$'s; and
$(\varphi_i\mathbb{E}\psi_i)_{i=1}^m:=\Box(F_0\lor\cdots\lor F_m)$
says that at every outcome with nonzero probability, the number of
true $\varphi_k$'s is the same as the number of true $\psi_k$'s.
Using this notation, it is possible to express the fourth condition of
Scott's theorem and thereby obtain completeness for the probabilistic
interpretation.

In the present paper, we follow this tradition of studying probability
from a qualitative (i.e., non-numerical) point of view using modal
logic.  However, our focus shall not be on the binary relation
$\preceq$ of qualitative probabilistic comparison but instead on the
unary notions of certainty (i.e., having probability $1$) and ``high''
probability (i.e., having a probability greater than some fixed
rational-number threshold $c\geq\frac 12$). That is, our interest is
in \emph{unary modal logics of high probability}.

For convenience in this study, we shall identify epistemic notions
with probabilistic assignment, which suggests a connection with
subjective probability \cite{Jeffrey2004:sptrt}.  In particular, we
identify knowledge with probabilistic certainty (i.e., probability
$1$) and belief with probability greater than some fixed
rational-number threshold $c\geq\frac 12$.  Therefore, instead of the
unary operator $\Box$, we shall use the unary operator $K$ and assign
this operator an epistemic reading: $K\varphi$ says that the agent
knows $\varphi$, which means she assigns $\varphi$ subjective
probability $1$.  We shall use the unary modal operator $B$ to express
belief: $B\varphi$ says that the agent believes $\varphi$, which means
she assigns $\varphi$ a subjective probability exceeding the threshold
$c$ (which will always be a fixed value within a given context or
theory).  Though our readings of these formulas are epistemic and
doxastic, we stress that our technical results are independent of this
reading, so someone who disagrees with subjective probability or our
epistemic/doxastic readings is encouraged to think of our work purely
in terms of high probability: $K\varphi$ says $P(\varphi)=1$, and
$B\varphi$ says $P(\varphi)>c$ for some fixed
$c\in[\frac 12,1)\cap\Rat$. That is, the technical results of our work
are in no way dependent on our use of epistemic/doxastic notions or on
the philosophy of subjective probability.

Lenzen \cite{Lenzen2003:kbasp,Lenzen1980:gwuw} is to our knowledge the
first to consider a modal logic of high probability for the threshold
$c=\frac 12$.  Actually, his perspective is slightly different than
the one we adopt here.  First, his reading of formulas is different
(though not in any deep way): he identifies ``the agent is convinced
of $\varphi$'' with $P(\varphi)=1$ and ``$\psi$ is believed'' by
$P(\psi)>\frac 12$.  More substantially, Lenzen's conviction (German:
\emph{\"{U}berzeuging}) does not imply truth.  Technically, this
amounts to permitting the possibility that there are outcomes having
probability zero.  For reasons of personal preference, we forbid this
in our study here, though this difference is non-essential, as it is
completely trivial from the technical perspective to change our
setting to allow zero-probability outcomes or to change Lenzen's
setting to forbid them. Therefore, we credit Lenzen's work
\cite{Lenzen1980:gwuw} as the first to provide a proof of
probabilistic completeness for $c=\frac 12$.  As with Segerberg's and
G{\"a}rdenfors' probabilistic completeness results, Lenzen's proof
made crucial use of Scott's work.

In more recent work, Herzig \cite{Herzig2003:mpbaa} considered a logic
of belief and action in which belief in $\varphi$ is identified with
$P(\varphi)>P(\lnot \varphi)$. This is equivalent to Lenzen's notion,
though Herzig does not study completeness.  Another recent work by
Kyberg and Teng \cite{KyburgTeng2012:tlorkr} investigated a notion of
``acceptance'' in which $\varphi$ is accepted whenever the probability
of $\lnot\varphi$ is at most some small $\epsilon$.  This gives rise
to the minimal modal logic $\mathsf{EMN}$, which is different than
Lenzen's logic.

We herein consider belief \`{a} la Lenzen not only for the case
$c=\frac 12$ but also for the case $c>\frac 12$.  As it turns out, the
logics for these cases are different, though our focus will be on the
logic for $c=\frac 12$ because this is the only threshold for which a
probability completeness result is known.  In particular, probability
completeness for $c>\frac 12$ is still open.  Thresholds $c<\frac 12$
permit simultaneous belief of $\varphi$ and $\lnot\varphi$ while
avoiding belief of any self-contradictory sentence such as the
propositional constant $\bot$ for falsehood. This might suggest some
connection with paraconsistent logic. However, we leave these logics
of low probability for future work, though we shall say a few words
more about them later in this paper.

In Section~\ref{Section:EPL}, we identify a Kripke-style semantics for
probability logic similar to
\cite{EijckSchwarzentruber2014:epls,Halpern2003:rau} (and no doubt to
many others).  We require that all worlds are probabilistically
possible but not necessarily epistemically so, and we provide some
examples of how this semantics works.  In particular, we demonstrate
that our requirement is not problematic: world $v$ can be made to have
probability zero relative to world $w$ if we cut the epistemic
accessibility relation between these worlds.

In Section~\ref{Section:CB}, we define our modal notions of certain
knowledge and of belief exceeding threshold $c$, explain the motto
``belief is willingness to bet,'' and prove a number of properties of
certain knowledge and this ``betting'' belief.  For instance, we show
that knowledge is $\mathsf{S5}$ and belief is not normal.  We show a
number of other threshold-specific properties of betting belief as
well.  In particular, we see that the belief modality extends the
minimal modal logic $\mathsf{EMND45}+\lnot B\bot$ by way of certain
schemes relating knowledge and belief.

We then introduce a formal modal language in
Section~\ref{Section:ENM}, relate this language to the probabilistic
notions of belief and knowledge, and introduce an epistemic
neighborhood semantics for the language.  We study the relationship
between the neighborhood and probabilistic semantics. In particular,
we introduce a notion of ``agreement'' between epistemic probability
models and epistemic neighborhood models, the key component of which
is this: an event $X$ is a neighborhood of a world $w$ if and only if
the probability measure $P_w$ at $w$ satisfies $P_w(X)>c$.  We use one
of Scott's theorems to prove that epistemic neighborhood models
satisfying certain properties give rise to agreeing epistemic
probability models for the threshold $c=\frac 12$.  This result we
credit to Lenzen; however, we prove this result anew in a modern,
streamlined form that we hope will make it more accessible.  The main
remaining open problem is to prove the analogous result for thresholds
$c\neq\frac 12$ (i.e., find the additional sufficient conditions on
epistemic neighborhood models we need to impose so as to guarantee the
existence of an agreeing epistemic probability model for threshold
$c\neq\frac 12$).  Finally, we prove that epistemic probability models
always give rise to agreeing epistemic neighborhood models.

In Section~\ref{Section:Calculi}, we introduce a basic modal theory
$\KB$ that is probabilistically sound.  We adapt an example due to
Walley and Fine \cite{WalleyFine1979:vomacp} that shows $\KB$ is
probabilistically incomplete.  This leads us to add additional
principles to $\KB$, thereby producing the modal theory $\KBeq$, our
name for our modern reformulation of Lenzen's modal theory of
knowledge and belief (or, in Lenzen's terminology, his theory of
``acceptance'' and belief).  Using the results from
Section~\ref{Section:ENM}, we prove that this logic is sound and
complete for epistemic probability models using threshold
$c=\frac 12$. Regarding the semantics based on our epistemic
neighborhood models, we prove that $\KB$ is sound and complete for the
full class of these models and that $\KBeq$ is sound and complete for
the smaller class that satisfies the additional Lenzen-derivative
properties needed to guarantee the existence of an agreeing
probability measure for threshold $c=\frac 12$.

Stated in an analogy: $\KB$ is to de Finetti's axiomatization as
$\KBeq$ is to the KPS/Scott axiomatization.  However, do not be
misled: de Finetti, KPS, and Scott considered qualitative
probabilistic comparison, which is a binary notion based on a binary
operator $\preceq$.  See also \cite{HollidayIcard2013:msaqs} for a
revival of this tradition. We, on the other hand, consider high
probability, which is a unary notion based on unary operators we
denote as $K$ and $B$.

Another version of our main open question can be restated in the
following syntactic form: given a threshold $c\neq\frac 12$, find the
additional principles that must be added to our probabilistically
sound but incomplete base logic $\KB$ in order to obtain a
probabilistically sound and complete logic for threshold $c$.  In our
conclusion, we present some additional sound principles that might
come up in this work, but we have not been able to find the
probabilistically sound and complete axiomatization for thresholds
$c\neq\frac 12$.

Given the link between epistemic neighborhood models and epistemic
probability models, our results may be viewed as a contribution to the
study connecting two schools of rational decision making: the
probabilist (e.g., \cite{koerner2008naive}) and the AI-based (e.g.,
\cite{KyburgTeng2012:tlorkr}).  We also hope that it will be of some
use in future work on qualitative probability.

\section{Epistemic Probability Models}
\label{Section:EPL} 

\begin{definition}
  \label{definition:epistemic-probability-model}
  We fix a set $\Prop$ of propositional letters.  An \emph{epistemic
    probability model} is a structure $\M=(W,R,V,P)$ satisfying the
  following.
  \begin{itemize} 
  \item $(W,R,V)$ is a finite single-agent $\mathsf{S5}$ Kripke model:
    \begin{itemize}
    \item $W$ is a finite nonempty set of ``worlds'' or ``outcomes.''
      An \emph{event} is a set $X\subseteq W$ of worlds.  When
      convenient, we identify a world $w$ with the singleton event
      $\{w\}$.
      
    \item $R\subseteq W\times W$ is an equivalence relation $R$ on
      $W$.  We let
      \[
      [w]:=\{v\in W\mid wRv\}
      \]
      denote the equivalence class of world $w$.  This is the set of
      worlds that agent cannot distinguish from $w$.

    \item $V:W\to\wp(\Prop)$ assigns a set $V(w)$ of propositional
      letters to each world $w\in W$.
    \end{itemize}

  \item $P:\wp(W)\to[0,1]$ is a probability measure over the finite
    algebra $\wp(W)$ satisfying the property of \emph{full support\/}:
    $P(w)\neq0$ for each $w\in W$.
  \end{itemize}
  A \emph{pointed epistemic probability model} is a pair $(\M,w)$
  consisting of an epistemic probability model $\M=(W,R,V,P)$ and
  world $w\in W$ called the \emph{point}.
\end{definition}

The agent's uncertainty as to which world is the actual world is given
by the equivalence relation $R$.  If $w$ is the actual world, then the
probability the agent assigns to an event $X$ at $w$ is given by
\begin{equation}
  P_w(X) := \frac{P(X\cap[w])}{P([w])}\enspace.
  \label{eq:probability}
\end{equation}
In words: the probability the agent assigns to event $X$ at world $w$
is the probability she assigns to $X$ conditional on her knowledge at
$w$. Slogan: subjective probability is always conditioned, and the
most general condition is given by the knowledge of the agent. This
makes sense because the right side of \eqref{eq:probability} is just
$P(X|[w])$, the probability of $X$ conditional on $[w]$.  Note that
$P_w(X)$ is always well-defined: we have $w\in[w]$ by the reflexivity
of $R$ and hence $0<P(w)\leq P([w])$ by full support, so the
denominator on the right side of \eqref{eq:probability} is nonzero.

\begin{example}[Horse racing] 
  \label{ExampleHorseRacing}
  Three horses compete in a race.  For each $i\in\{1,2,3\}$, horse
  $h_i$ wins the race in world $w_i$.  The agent can distinguish
  between these three possibilities, and she assigns the horses
  winning chances of $3{:}2{:}1$.  We represent this situation in the
  form of an epistemic probability model
  $\M_{\ref{ExampleHorseRacing}}$ pictured as follows:
  \begin{center}
    \begin{mytikz}
      \node[w,label={below:$w_1$}] (w1) {$h_1$};

      \node[w,right of=w1,label={below:$w_2$}] (w2) {$h_2$};

      \node[w,right of=w2,label={below:$w_3$}] (w3) {$h_3$};

      \path (w1) edge[<->] node{} (w2);

      \path (w2) edge[<->] node{} (w3);
    \end{mytikz}
    \[
      P = \textstyle\{w_1:\frac 36, w_2:\frac 26, w_3:\frac 16\}
    \]

    $\M_{\ref{ExampleHorseRacing}}$
  \end{center}
  When we picture epistemic probability models, the arrows of the
  agent are to be closed under reflexivity and transitivity.  With
  this convention in place, it is not difficult to verify that
  $P_{w_1} (\{w_1,w_3\}) = \frac 23$; that is, at $w_1$, the assigns
  probability $\frac 23$ to the event that the winner is horse $1$ or
  horse $3$.
\end{example} 

The property of full support says that each world is probabilistically
possible.  Therefore, in order to represent a situation in which the
agent is certain that horse $3$ can never win, we simply make the
$h_3$-worlds inaccessible via $R$.

\begin{example}[Certainty of impossibility]
  \label{ExampleHorseRacing2}
  We modify Example~\ref{ExampleHorseRacing} by eliminating the arrow
  between worlds $w_2$ and $w_3$.
  \begin{center}
    \begin{mytikz}
      \node[w,label={below:$w_1$}] (w1) {$h_1$};

      \node[w,right of=w1,label={below:$w_2$}] (w2) {$h_2$};

      \node[w,right of=w2,label={below:$w_3$}] (w3) {$h_3$};

      \path (w1) edge[<->] node{} (w2);
    \end{mytikz}
    \[
      P = \textstyle\{w_1:\frac 36, w_2:\frac 26, w_3:\frac 16\}
    \]

    $\M_{\ref{ExampleHorseRacing2}}$
  \end{center}
  At world $w_1$ in this picture, there is no accessible world at
  which horse $3$ wins.  Therefore, at world $w_1$, the agent assigns
  probability $0$ to the event that horse $3$ wins: $P_{w_1}(w_3)=0$.
\end{example}

We define a language $\Lang$ for reasoning about epistemic probability
models.

\begin{definition}
  The language $\Lang$ of \emph{(single-agent) probability logic} is
  defined by the following grammar.
  \begin{eqnarray*}
    \varphi & ::= & 
    \top \mid p \mid \neg\varphi \mid \varphi\land\varphi \mid
    t\geq 0
    \\
    t   & ::= & q \mid q \cdot P(\varphi) \mid t+t
    \\
    &&
    \text{\footnotesize 
      $p\in\Prop$,
      $q\in\Rat$
    }
  \end{eqnarray*}
  We adopt the usual abbreviations for Boolean connectives.  We define
  the relational symbols $\leq$, $>$, $<$, and $=$ in terms of $\geq$
  as usual.  For example, $t=s$ abbreviates $(t\geq s)\land(s\geq t)$.
  We also use the obvious abbreviations for writing linear
  inequalities.  For example, $P(p)\leq 1-q$ abbreviates
  $1+(-q)+(-1)\cdot P(p)\geq 0$.
\end{definition}

\begin{definition} 
  Let $\M=(W,R,V,P)$ be an epistemic probability model.  We define
  a binary truth relation $\modelsp$ between a pointed epistemic
  probability model $(\M,w)$ and $\Lang$-formulas as follows.
  \[
  \renewcommand{\arraystretch}{1.3}
  \begin{array}{lcl}
    \M,w\modelsp\top 
    \\
    \M,w\modelsp p & \text{iff} & 
    p \in V(w) 
    \\
    \M,w\modelsp\neg\varphi & \text{iff} &
    \M,w\not\modelsp\varphi
    \\
    \M,w\modelsp\varphi\land\psi & \text{iff} &
    \M,w\modelsp\varphi \text{ and } \M,w\modelsp\psi
    \\
    \M,w\modelsp t\geq 0 & \text{iff} &
    \sem{t}_w\geq 0
  \end{array}
  \]
  \begin{eqnarray*}
    \semp{\varphi} & := &
    \{ u \in W\mid \M,u\modelsp\varphi \}
    \\
    P_w(X) & := & 
    \displaystyle\frac{P(X\cap[w])}{P([w])}
    \\
    \sem{q}_w & := & q
    \\
    \sem{q\cdot P(\varphi)}_w & := & 
    q\cdot P_w(\semp{\varphi}) 
    \\
    \sem{t+t'}_w & := &
    \sem{t}_w + \sem{t'}_w
  \end{eqnarray*}
  Validity of $\varphi\in\Lang$ in epistemic probability model $\M$,
  written $\M\modelsp\varphi$, means that $\M,w\modelsp\varphi$ for each
  world $w\in W$.  Validity of $\varphi\in\Lang$, written $\modelsp\varphi$,
  means that $\M\modelsp\varphi$ for each epistemic probability model
  $\M$.
\end{definition} 

\section{Certainty and Belief} 
\label{Section:CB} 

\cite{Eijck2013:lap} formulates and proves a ``certainty theorem''
relating certainty in epistemic probability models to knowledge in a
version of these models in which the probabilistic information is
removed.  This motivates the following definition.

\begin{definition}[Knowledge as Certainty]
  We adopt the following abbreviations.
  \begin{itemize}
  \item $K\varphi$ abbreviates $P(\varphi)=1$. 

    We read $K\varphi$ as ``the agent knows $\varphi$.''

  \item $\check K\varphi$ abbreviates $\lnot K\lnot\varphi$.

    We read $\check K\varphi$ as ``$\varphi$ is consistent with the agent's
    knowledge.''
  \end{itemize}
\end{definition}

\begin{theorem}[\cite{Eijck2013:lap}]
  \label{theorem:knowledge}
  $K$ is an $\mathsf{S5}$ modal operator:
  \begin{enumerate}
  \item $\modelsp \varphi$ for each $\Lang$-instance $\varphi$ of a scheme
    of classical propositional logic.

    Axioms of classical propositional logic are valid.

  \item $\modelsp K(\varphi\to\psi)\to(K\varphi\to K\psi)$
    
    Knowledge is closed under logical consequence.

  \item $\modelsp K\varphi\to \varphi$

    Knowledge is veridical.
    
  \item $\modelsp K\varphi\to KK\varphi$

    Knowledge is positive introspective:  it is known what is known.
    
  \item $\modelsp \lnot K\varphi\to K\lnot K\varphi$

    Knowledge is negative introspective: it is known what is not
    known.
    
  \item $\modelsp\varphi$ implies $\modelsp K\varphi$

    All validities are known.

  \item $\modelsp\varphi\to\psi$ and $\modelsp\varphi$ together imply
    $\modelsp\psi$.

    Validities are closed under the rule of Modus Ponens.
  \end{enumerate}
\end{theorem}

We define belief in a proposition $\varphi$ as willingness to take bets
on $\varphi$ with the odds being better than some rational number
$c\in(0,1)\cap\Rat$.  This leads to a number of degrees of belief, one
for each threshold $c$.

\begin{definition}[Belief as Willingness to Bet]
  \label{definition:belief}
  Fix a threshold $c\in(0,1)\cap\Rat$.
  \begin{itemize}
  \item $B^c\varphi$ abbreviates $P(\varphi)>c$.

    We read $B^c\varphi$ as ``the agent believes $\varphi$ with threshold
    $c$.''

  \item $\check B^c\varphi$ abbreviates $\lnot B^c\lnot\varphi$.

    We read $\check B\varphi$ as ``$\varphi$ is consistent with the agent's
    threshold-$c$ beliefs.''
  \end{itemize}
  If the threshold $c$ is omitted (either in the notations $B^c\varphi$
  and $\check B^c\varphi$ or in the informal readings of these
  notations), it is assumed that $c=\frac 12$.
\end{definition}

This notion of belief comes from subjective probability
\cite{Jeffrey2004:sptrt}.  In particular, fix a threshold
$c=p/q\in(0,1)\cap\Rat$.  Suppose that the agent believes $\varphi$
with threshold $c=p/q$; that is, $P(\varphi)>p/q$.  If the agent
wagers $p$ dollars for a chance to win $q$ dollars on a bet that
$\varphi$ is true, then she expects her net winnings to be
\[
[(q-p)\cdot P(\varphi)] - [p\cdot(1-P(\varphi))] = q\cdot P(\varphi) - p
\]
dollars on this bet.  This is a positive number of dollars if and only
if $q\cdot P(\varphi)>p$.  But notice that the latter is guaranteed by
the assumption $P(\varphi)>p/q$.  Therefore, it is rational for the
agent to take this bet.  Said in the parlance of the subjective
probability literature: ``If the agent stakes $p$ to win $q$ in a bet
on $\varphi$, then her winning expectation is positive in case she
believes $\varphi$ with threshold $c= p/q$.''  Or in a short motto:
``Belief is willingness to bet.''

\begin{remark}
  \label{remark:low-threshold}
  Belief based on threshold $c=0$ or $c=1$ is trivial to express in
  terms of negation, $K$, and falsehood $\bot$.  So we do not consider
  these thresholds here.  Beliefs based on low-thresholds
  $c\in(0,\frac 12)\cap\Rat$ have unintuitive and unusual features.
  First, low-threshold beliefs unintuitively permit inconsistency of
  the kind that an agent can believe both $\varphi$ and $\lnot\varphi$
  while avoiding inconsistency of the kind that the agent can believe
  a self-contradictory formula such as $\bot$. (This suggests some
  connection with paraconsistent logic.)  Second, the dual of a
  low-threshold belief implies the belief at that threshold (i.e.,
  $\check B^c\varphi\to B^c\varphi$), which is unusual if we assign
  the usual ``consistency'' reading to dual operators (i.e.,
  ``$\varphi$ is consistent with the agent's beliefs implies $\varphi$
  is believed'' is unusual).  Since low-threshold
  $c\in(0,\frac12)\cap\Rat$ beliefs have these unintuitive and unusual
  features, we leave their study for future work, focusing instead on
  thresholds $c\in[\frac 12,1)\cap\Rat$.
\end{remark}

The following lemma provides a useful characterization of the dual
$\check B^c\varphi$.

\begin{lemma}
  \label{lemma:dual}
  Let $\M=(W,R,V,P)$ be an epistemic probability model.
  \begin{enumerate}
  \item \label{item:dual} $\M,w\modelsp\check B^c\varphi$ iff
    $\M,w\modelsp P(\varphi)\geq 1-c$.

  \item \label{item:dual-half} $\M,w\modelsp\check B^{\frac 12}\varphi$ iff
    $\M,w\modelsp P(\varphi)\geq \frac 12$.

  \end{enumerate}
\end{lemma}
\begin{proof}
  For Item~\ref{item:dual}, we have the following:
  \[
  \renewcommand{\arraystretch}{1.3}
  \begin{array}{lll}
    &
    \M,w\modelsp\check B^c\varphi
    \\
    \text{iff} &
    \M,w\modelsp\lnot B^c\lnot\varphi 
    & \text{by definition of $\check B^c\varphi$}
    \\
    \text{iff} &
    P_w(\semp{\lnot\varphi})\not>c
    & \text{by definition of $B^c\varphi$ and $\modelsp$}
    \\
    \text{iff} &
    P_w(\semp{\lnot\varphi})\leq c
    & \text{since $\Rat$ is totally ordered}
    \\
    \text{iff} &
    P_w(\semp{\varphi})\geq 1-c
    & \text{since $\semp{\lnot\varphi}=W-\semp{\varphi}$}
  \end{array}
  \]
  For Item~\ref{item:dual-half}, apply Item~\ref{item:dual} 
  with $c =\frac 12$. 
\end{proof}


We now consider a simple example.

\begin{example}[Non-normality]
  \label{example:non-normality}
  In this variation, all horses have equal chances of winning and the
  agent knows this.
  
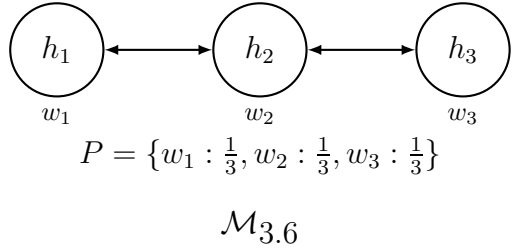
\begin{center}
    \begin{mytikz}
      \node[w,label={below:$w_1$}] (w1) {$h_1$};

      \node[w,right of=w1,label={below:$w_2$}] (w2) {$h_2$};

      \node[w,right of=w2,label={below:$w_3$}] (w3) {$h_3$};

      \path (w1) edge[<->] node{} (w2);

      \path (w2) edge[<->] node{} (w3);
    \end{mytikz}

    $P=\{w_1:\frac 13, w_2:\frac 13, w_3:\frac 13$\}

    \bigskip
    $\M_{\ref{example:non-normality}}$
  \end{center}
  Recalling that an omitted threshold $c$ is implicitly assumed to be
  $\frac 12$, the following are readily verified.
  \begin{enumerate}
  \item
    $\M_{\ref{example:non-normality}}\modelsp B(h_1\lor h_2\lor h_3)$.

    The agent believes the winning horse is among the three.

    (The agent is willing to bet that the winning horse is among the
    three.)

  \item
    $\M_{\ref{example:non-normality}}\modelsp B(h_1\lor h_2)\land
    B(h_1\lor h_3)\land B(h_2\lor h_3)$.

    The agent believes the winning horse is among any two.

    (The agent is willing to bet that the winning horse is among any
    two.)

  \item \label{item:conjuncts}
    $\M_{\ref{example:non-normality}}\modelsp B_a\lnot h_1\land
    B_a\lnot h_2\land B_a\lnot h_3$.

    The agent believes the winning horse is not any particular one.

    (The agent is willing to bet that the winning horse is not any
    particular one.)

  \item \label{item:conjunction}
    $\M_{\ref{example:non-normality}}\modelsp \lnot B(\lnot h_1\land \lnot
    h_2)$.

    The agent does not believe that both horses $1$ and $2$ do not
    win.

    (The agent is not willing to bet that both horses $1$ and $2$ do
    not win.)
  \end{enumerate}
\end{example} 

It follows from Items~\ref{item:conjuncts} and \ref{item:conjunction}
of Example~\ref{example:non-normality} that the present notion of
belief is not closed under conjunction.  This is discussed as part of
the literature on the ``Lottery Paradox''
\cite{Kyburg1961:patlorb}.\footnote{The usual formulation of the
  Lottery Paradox: it is paradoxical for an agent to believe that one
  of $n$ lottery tickets will be a winner (i.e., ``some ticket is a
  winner'') without believing of any particular ticket that it is the
  winner (i.e., ``for each $i\in\{1,\dots,n\}$, ticket $i$ is not a
  winner'').}  However, there is no reason in general that it is
paradoxical to assign a conjunction $\varphi \land \psi$ a lower
probability than either of its conjunctions.  Indeed, if $\varphi$ and
$\psi$ are independent, then the probability of their conjunction
equals the product of their probabilities, so unless one of $\varphi$ or
$\psi$ is certain or impossible, the probability of $\varphi \land \psi$
will be less than the probability of $\varphi$ and less than the
probability of $\psi$.

We set aside philosophical arguments for or against closure of belief
under conjunction and instead turn our attention to the study of the
properties of the present notion of belief.  One of these is a
complicated but useful property due to Scott \cite{Sco64:JMP} that
makes use of notation due to Segerberg \cite{Segerberg1971:qpiams}.

\begin{definition}[Segerberg notation; \cite{Segerberg1971:qpiams}]
  \label{definition:segerberg-notation}
  Fix a positive integer $m\in\Int^+$ and formulas
  $\varphi_1,\dots,\varphi_m$ and $\psi_1,\dots,\psi_m$.  The expression
  \begin{equation}
    (\varphi_1,\dots,\varphi_m\mathbb{I}\psi_1,\dots,\psi_m)
    \label{eq:segerberg}
  \end{equation}
  abbreviates the formula
  \[
  K(F_0\lor F_1\lor F_2 \lor \cdots \lor F_m)\enspace,
  \]
  where $F_i$ is the disjunction of all conjunctions
  \[
  d_1\varphi_1\land\cdots\land d_m\varphi_m \land
  e_1\psi_1\land\cdots\land e_m\psi_m
  \]
  satisfying the property that \emph{exactly} $i$ of the $d_k$'s are
  the empty string, \emph{at least} $i$ of the $e_k$'s are the empty
  string, and the rest of the $d_k$'s and $e_k$'s are the negation
  sign $\lnot$.  We may write $(\varphi_i\mathbb{I}\psi_i)_{i=1}^m$ as
  an abbreviation for \eqref{eq:segerberg}.  Finally, let
  \[
  (\varphi_i\mathbb{E}\psi_i)_{i=1}^m
  \quad\text{abbreviate}\quad
  (\varphi_i\mathbb{I}\psi_i)_{i=1}^m\land(\psi_i\mathbb{I}\varphi_i)_{i=1}^m
  \enspace.
  \]
  We also allow the use of $\mathbb{E}$ in a notation similar to
  \eqref{eq:segerberg}.
\end{definition}

The formula $(\varphi_i\mathbb{I}\psi_i)_{i=1}^m$ says that the agent
knows that the number of true $\varphi_i$'s is less than or equal to the
number of true $\psi_i$'s.  Put another way,
$(\varphi_i\mathbb{I}\psi_i)_{i=1}^m$ is true if and only if every one of
the agent's epistemically accessible worlds satisfies at least as many
$\psi_i$'s as $\varphi_i$'s.  The formula
$(\varphi_i\mathbb{E}\psi_i)_{i=1}^m$ says that every one of the agent's
epistemically accessible worlds satisfies exactly as many $\psi_i$'s
as $\varphi_i$'s.

\begin{definition}[Scott scheme; \cite{Sco64:JMP}]
  \label{definition:scott-schemes}
  We define the following scheme:
  \begin{equation}\tag{Scott}
    \textstyle [
    (\varphi_i\mathbb{I}\psi_i)_{i=1}^m
    \land B^c \varphi_1 \land \bigwedge_{i=2}^m \check B^c \varphi_i] \to
    \bigvee_{i=1}^m B^c\psi_i
  \end{equation}
  If $m=1$, then $\bigwedge_{i=2}^m \check B^c\varphi_i$ is $\top$.  Note
  that (Scott) is meant to encompass the indicated scheme for each
  positive integer $m\in\Int^+$.
\end{definition}

(Scott) says that if the agent knows the number of true $\varphi_i$'s is
less than or equal to the number of true $\psi_i$'s, she believes
$\varphi_1$ with threshold $c$, and the remaining $\varphi_i$'s are each
consistent with her threshold-$c$ beliefs, then she believes one of
the $\psi_i$'s with threshold $c$.  Adapting a proof of Segerberg
\cite{Segerberg1971:qpiams}, we show that belief with threshold
$c=\frac12$ satisfies (Scott).

We report this result along with a number of other properties
in the following proposition.

\begin{theorem}[Properties of Belief]
  \label{theorem:belief}
  For $c\in(0,1)\cap\Rat$, we have:
  \begin{enumerate}
  \item \label{item:B-not-normal}
    $\not\modelsp B^c(\varphi\to\psi)\to(B^c\varphi\to B^c\psi)$.

    Belief is not closed under logical consequence.

    (So $B^c$ is not a normal modal operator.)

  \item \label{item:B-not-T} $\not\modelsp B^c\varphi\to\varphi$.

    Belief is not veridical.

  \item \label{item:B-C} $\modelsp K\varphi\to B^c\varphi$.

    What is known is believed.

  \item \label{item:B-B} $\modelsp\lnot B^c\bot$.

    The propositional constant $\bot$ for falsehood is not believed.

  \item \label{item:B-N} $\modelsp B^c\top$.

    The propositional constant $\top$ for truth is believed.

  \item \label{item:B-Ap} $\modelsp B^c\varphi\to KB^c\varphi$.

    What is believed is known to be believed.

  \item \label{item:B-An} $\modelsp \lnot B^c\varphi\to K\lnot B^c\varphi$.

    What is not believed is known to be not believed.

  \item \label{item:B-M}
    $\modelsp K(\varphi\to\psi)\to(B^c\varphi\to B^c\psi)$.

    Belief is closed under known logical consequence.

  \item \label{item:B-D} If $c\in[\frac 12,1)$, then
    $\modelsp B^c\varphi\to \check B^c\varphi$.

    High-threshold belief is consistent: belief in $\varphi$ implies
    disbelief in $\lnot\varphi$.

  \item \label{item:B-SC}
    $\modelsp \check{B}^{\frac 12} \varphi \land \check{K}(\neg \varphi
    \land \psi) \rightarrow B^{\frac 12} (\varphi \lor \psi)$.
    
    For mid-threshold belief, if $\varphi$ is consistent with the agent's
    beliefs and $\lnot\varphi\land\psi$ is consistent with her knowledge,
    then she believes $\varphi\lor\psi$.

  \item \label{item:B-Len}
    $ \textstyle \modelsp [(\varphi_i\mathbb{I}\psi_i)_{i=1}^m \land
    B^{\frac 12}\varphi_1 \land \bigwedge_{i=2}^m \check B^{\frac
      12}\varphi_i] \to \bigvee_{i=1}^m B^{\frac 12}\psi_i $.

    Mid-threshold belief satisfies (Scott).
  \end{enumerate}
\end{theorem}
\begin{proof}
  We consider each item in turn.
  \begin{enumerate}
  \item Given $c\in(0,1)\cap\Rat$ and integers $p$ and $q$ such that
    $p/q=c$, we define $\M$ as the modification of the model
    $\M_{\ref{example:non-normality}}$ of
    Example~\ref{example:non-normality} obtained by changing $P$ as
    follows:
    \[
    P := \left\{ w_1:\frac{q-p}{2q},\; w_2:\frac pq,\;
      w_3:\frac{q-p}{2q} \right\}\enspace.
    \]
    Since $0<p<q$, it follows that
    \begin{eqnarray*}
      P_{w_1}(\semp{\lnot h_1\to h_2}) &=&
      P_{w_1}(\{w_1,w_2\}) =
      \frac{q+p}{2q}>\frac {2p}{2q}=\frac pq,
      \\
      P_{w_1}(\semp{\lnot h_1}) &=&
      P_{w_1}(\{w_2,w_3\}) = \frac{q+p}{2q}>\frac {2p}{2q}
                                    =\frac pq,\text{ and}
      \\
      P_{w_1}(\semp{h_2})&=& 
      P_{w_1}(w_2) = \frac pq\enspace.
    \end{eqnarray*}
    Therefore, we have
    \[
    \M,w_1\modelsp B^c(\lnot h_1\to h_2)\land B^c\lnot h_1\land \lnot
    B^ch_2\enspace.
    \]
    
  \item For $\M$ defined in the proof of Item~\ref{item:B-not-normal},
    we have
    \[
    \M,w_1\modelsp h_1\land B^c\lnot h_1\enspace.
    \]
    
  \item $\M,w\modelsp K\varphi$ implies $P_w(\semp{\varphi})=1 > c$.
    Hence $\M,w\modelsp B^c\varphi$.

  \item $P_w(\semp{\bot})=0<c$.  Hence
    $\M,w\models\lnot B^c\bot$.

  \item $P_w(\semp{\top})=1>c$. Hence $\M,w\modelsp B^c\top$.

  \item $\M,w\modelsp B^c\varphi$ implies $P_w(\semp{\varphi})>c$.  To show
    that $\M,w\modelsp KB^c\varphi$, we must prove that
    \[
    P_w(\semp{B^c\varphi})=
    \frac{P(\semp{B^c\varphi}\cap[w])}{P([w])} =1\enspace.
    \]
    To show this, we prove that $\semp{B^c\varphi}\cap[w]=[w]$.  So
    choose $u\in[w]$.  Since $R$ is an equivalence relation, we have
    \begin{equation*}
      P_u(\semp{\varphi})=
      \frac{P(\semp{\varphi}\cap[u])}{P([u])}=
      \frac{P(\semp{\varphi}\cap[w])}{P([w])}=
      P_w(\semp{\varphi})>c\enspace,
    \end{equation*}
    which implies $u\in\semp{B^c\varphi}$.  The result follows.
    
  \item The argument is similar to that for Item~\ref{item:B-Ap},
    though we note that $\M,w\modelsp\lnot B^c\varphi$ implies
    $P_w(\semp{\varphi})\leq c$.
    
  \item We assume that $\M,w\modelsp K(\varphi\to\psi)$ and
    $\M,w\modelsp B^c\varphi$.  This means that
    $P_w(\semp{\varphi\to\psi})=1$ and $P_w(\semp{\varphi})>c$.  But then it
    follows that $P_w(\semp{\psi})>c$ as well, which is what it means
    to have $\M,w\modelsp B^c\psi$.


  \item Assume $c\in[\frac 12,1)\cap\Rat$ and $\M,w\modelsp B^c\varphi$.
    Then $P_w(\semp{\varphi})>c \geq 1 - c$. So
    $P_w(\semp{\varphi})\geq 1-c$.  The result therefore follows by
    Lemma~\ref{lemma:dual}.

  \item We prove something more general.  Assume
    $c\in(0,\frac 12]\cap\Rat$ and $\M,w\modelsp\check B^c\varphi$.  By
    Lemma~\ref{lemma:dual}, it follows that $P_w(\semp{\varphi})\geq c$.
    Let us assume further that
    $\M,w\modelsp\check K(\lnot\varphi\land\psi)$.  This means
    \[
    1\neq P_w(\semp{\lnot(\lnot\varphi\land\psi)})=
    \frac{P(\semp{\lnot(\lnot\varphi\land\psi)}\cap[w])}{P([w])}
    \enspace,
    \]
    which implies there exists
    $v\in \semp{\lnot\varphi\land\psi}\cap[w]$.  Since $P(v)>0$ by
    full support, it follows that
    \begin{eqnarray*}
    P_w(\semp{\varphi\lor\psi}) &=&
    \frac{ P(\semp{\varphi\lor\psi}\cap [w]) }{ P([w]) }
    \\
    &=&
    \frac{ P(\semp{\varphi}\cap [w]) }{ P([w]) } +
    \frac{ P(\semp{\lnot\varphi\land\psi}\cap [w]) }{ P([w]) }
    \\
    &\geq&
    \frac{ P(\semp{\varphi}\cap[w]) }{P([w])} + \frac{ P(v) }{P([w])}
    \\
    &=&
    P_w(\semp{\varphi})+\frac{ P(v) }{P([w])}
    \\
    &\geq&
    c + \frac{ P(v) }{P([w])}
    > c\enspace.
    \end{eqnarray*}
    That is,
    $\M,w\modelsp B^c(\varphi\lor\psi)$.

  \item Again, we prove something more general.  We assume
    $c\in(0,\frac 12]\cap\Rat$ plus the following:
    \begin{eqnarray}
      &&
      \M,w\modelsp (\varphi_i\mathbb{I}\psi_i)_{i=1}^m
      \label{eq:S}
      \\
      &&
      \M,w\modelsp B^c\varphi_1
      \label{eq:Bphi1}
      \\
      &&
      \textstyle \M,w\modelsp \bigwedge_{i=2}^m\check B^c\varphi_i
      \label{eq:checkB}
    \end{eqnarray}
    We recall the meaning of \eqref{eq:S}: for each $v\in[w]$, the
    number of $\varphi_i$'s true at $v$ is less than or equal to the
    number of $\psi_k$'s true at $v$.  It therefore follows from
    \eqref{eq:S} that
    \begin{equation}
      P_w(\semp{\varphi_1})+\cdots+P_w(\semp{\varphi_m})\leq
      P_w(\semp{\psi_1})+\cdots+P_w(\semp{\psi_m})\enspace.
      \label{eq:sum}
    \end{equation}
    Outlining an argument due to Segerberg
    \cite[pp.~344--346]{Segerberg1971:qpiams}, the reason for this is
    as follows: we think of each world $v\in[w]$ as being assigned a
    ``weight'' $P_w(v)$.  A member $P_w(\semp{\varphi_i})$ of the sum on
    the left of \eqref{eq:sum} is just a total of the weight of every
    $v\in[w]$ that satisfies $\varphi_i$; that is,
    \[
    P_w(\semp{\varphi_i})=\sum\{P_w(v)\mid
    v\in\semp{\varphi_i}\cap[w]\}\enspace.
    \]
    Assumption \eqref{eq:S} tells us that for each $v\in[w]$, the
    number of totals $P_w(\semp{\varphi_i})$ on the left of
    \eqref{eq:sum} to which $v$ contributes its weight is less than or
    equal to the number of totals $P_w(\semp{\psi_k})$ on the right of
    \eqref{eq:sum} to which $v$ contributes its weight.  But then the
    sum of totals on the left must be less than or equal to the sum of
    totals on the right.  Hence \eqref{eq:sum} follows.

    Having established \eqref{eq:sum}, we now proceed further with the
    overall proof.  By \eqref{eq:Bphi1}, we have
    $P_w(\semp{\varphi_1})>c$.  Applying \eqref{eq:checkB} and
    Lemma~\ref{lemma:dual}, we have $P_w(\varphi_i)\geq c$ for each
    $i\in\{2,\dots,m\}$.  Hence
    \[
    P_w(\semp{\psi_1})+\cdots+P_w(\semp{\psi_m})\geq
    P_w(\semp{\varphi_1})+\cdots+P_w(\semp{\varphi_m})> mc\enspace.
    \]
    That is, the sum of the $P_w(\semp{\psi_k})$'s must exceed $mc$.
    Since each member of this $m$-member sum is non-negative, it
    follows that at least one member must exceed $c$.  That is, there
    exists $j\in\{1,\dots,m\}$ such that $P_w(\semp{\psi_j})>c$.
    Hence $\M,w\modelsp\bigvee_{j=1}^mB^c\psi_j$. \qedhere
  \end{enumerate}
\end{proof}

\section{Epistemic Neighborhood Models}
\label{Section:ENM}

The modal formulas $K\varphi$ and $B^c\varphi$ were taken as abbreviations
in the language $\Lang$ of probability logic.  We wish to consider a
propositional modal language that has knowledge and belief operators
as primitives.

\begin{definition}
  The language $\Lang_\KB$ of \emph{(single-agent) knowledge and
    belief} is defined by the following grammar.
  \begin{eqnarray*}
    \varphi & ::= & 
    \top \mid p \mid \neg\varphi \mid \varphi\land\varphi \mid
    K\varphi \mid B\varphi
    \\
    &&
    \text{\footnotesize 
      $p\in\Prop$
    }
  \end{eqnarray*}
  We adopt the usual abbreviations for other Boolean connectives and
  define the dual operators $\check K:=\lnot K\lnot$ and
  $\check B:=\lnot B\lnot$.  Finally, the $\Lang_\KB$-formula
  \[
  (\varphi_1,\dots,\varphi_m\mathbb{I}\psi_1,\dots,\psi_m)
  \]
  and its abbreviation $(\varphi_i\mathbb{I}\psi_i)_{i=1}^m$ are given as
  in Definition~\ref{definition:segerberg-notation} except that all
  formulas are taken from the language $\Lang_\KB$.
\end{definition}

Our goal will be to develop a possible worlds semantics for
$\Lang_\KB$ that links with the probabilistic setting by making the
following translation truth-preserving.

\begin{definition}[Translation]
  \label{definition:translation}
  For $c\in(0,1)\cap\Rat$, we define $c: \Lang_\KB \to \Lang$ as
  follows.
  \[
  \renewcommand{\arraystretch}{1.2}
  \begin{array}{ccl@{\qquad}l}
    \top^c & := & \top
    \\
    p^c & := & p 
    \\
    (\neg \varphi)^c & := & \neg \varphi^c 
    \\
    (\varphi \land \psi)^c & := & \varphi^c \land \psi^c 
    \\
    (K\varphi)^c & := & P(\varphi^c) = 1
    & (= K\varphi^c \text{ in } \Lang)
    \\
    (B\varphi)^c & := & \textstyle P(\varphi^c) > c
    & (= B^c\varphi^c \text{ in } \Lang)
  \end{array}
  \]
\end{definition} 

Since we have seen that the probabilistic belief operator $B^c$ is not
a normal modal operator
(Theorem~\ref{theorem:belief}\eqref{item:B-not-normal}), we opt for a
neighborhood semantics for $\Lang_\KB$ \cite[Ch.~7]{Chellas:ml} with
an epistemic twist.

\begin{definition} 
  An \emph{epistemic neighborhood model} is a structure
  \[
  \M=(W,R,V,N)
  \]
  satisfying the following.
  \begin{itemize}
  \item $(W,R,V)$ is a finite single-agent $\mathsf{S5}$ Kripke model
    (as in Definition~\ref{definition:epistemic-probability-model}).
    As before, we let
    \[
    [w]:=\{v\in W\mid wRv\}
    \]
    denote the equivalence class of world $w$.  This is the set of
    worlds the agent cannot distinguish from $w$.

  \item $N : W \to \wp(\wp(W))$ is a \emph{neighborhood function} that
    assigns to each world $w\in W$ a collection $N(w)$ of sets of
    worlds---each such set called a \emph{neighborhood} of
    $w$---subject to the following conditions.
    \begin{description}
    \item[(kbc)] $\forall X \in N(w) : X \subseteq [w]$.

    \item[(kbf)] $\emptyset\notin N(w)$.
      
    \item[(n)] $[w]\in N(w)$.
      
    \item[(a)] $\forall v \in [w] : N(v) = N(w)$.

    \item[(kbm)] $\forall X \subseteq Y \subseteq [w] : 
      \text{ if } X \in N(w) \text{, then } Y \in N(w)$.
   \end{description}
  \end{itemize}
  A \emph{pointed epistemic neighborhood model} is a pair $(\M,w)$
  consisting of an epistemic neighborhood model $\M$ and a world $w$
  in $M$.
\end{definition}

An epistemic neighborhood model is a variation of a neighborhood model
that includes an epistemic component $R$.  Intuitively, $[w]$ is the
set of worlds the agent knows to be possible at $w$ and each
$X\in N(w)$ represents a proposition that the agent believes at $w$.
The condition that $R$ be an equivalence relation ensures that
knowledge is closed under logical consequence, veridical (i.e., only
true things can be known), positive introspective (i.e., the agent
knows what she knows), and negative introspective (i.e., the agent
knows what she does not know).

Property (kbc) ensures that the agent does not believe a proposition
$X\subseteq W$ that she knows to be false: if $X$ contains a world in
$w'\in(W-[w])$ that the agent knows is not possible with respect to
the actual world $w$, then she knows that $X$ cannot be the case and
hence she does not believe $X$.  Property (kbf) ensures that no
logical falsehood is believed, while Property (n) ensures that every
logical truth is believed.  Property (a) ensures that $X$ is believed
if and only if it is known that $X$ is believed. Property (kbm) says
that belief is monotonic: if an agent believes $X$, then she believes
all propositions $Y\supseteq X$ that follow from $X$.

We now turn to the definition of truth for the language $\Lang_\KB$.

\begin{definition} 
  Let $\M = (W,R,V,N)$ be an epistemic neighborhood model.  We define
  a binary truth relation $\modelsn$ between a pointed epistemic
  neighborhood model $(\M,w)$ and $\Lang_\KB$-formulas and a function
  $\semn{\cdot}^\M:\Lang_\KB\to \wp(W)$ as follows.
  \begin{eqnarray*} 
    \semn{\varphi}^\M & := & \{v\in W\mid \M,v\modelsn\varphi\}
    \\
    \M, w \modelsn p & \text{ iff } & p \in V(w) 
    \\
    \M, w \modelsn \neg \varphi & \text{ iff } & \M, w \not\modelsn \varphi 
    \\
    \M, w \modelsn \varphi\land\psi  & \text{ iff } 
    & \M, w \modelsn \varphi \text{ and } \M, w \modelsn \psi
    \\
    \M, w \modelsn K \varphi  & \text{ iff } & 
    [w]\subseteq\semn{\varphi}^\M
    \\
    \M, w \modelsn B \varphi  & \text{ iff } &
    [w]\cap \semn{\varphi}^\M \in N(w)
  \end{eqnarray*}
  Validity of $\varphi\in\Lang_\KB$ in an epistemic neighborhood model
  $\M$, written $\M\modelsn\varphi$, means that $\M,w\modelsn\varphi$ for
  each world $w\in W$.  Validity of $\varphi\in\Lang_\KB$, written
  $\modelsn\varphi$, means that $\M\modelsn\varphi$ for each epistemic
  neighborhood model $\M$.  For a class $\mathcal{C}$ of epistemic
  neighborhood models, we write $\mathcal{C}\modelsn\varphi$ to mean that
  $\M\modelsn\varphi$ for each $\M\in\mathcal{C}$.
\end{definition}

Intuitively, $K\varphi$ is true at $w$ iff $\varphi$ holds at all worlds
epistemically possible with respect to $w$, and $B\varphi$ holds at $w$
iff the epistemically possible $\varphi$-worlds make up a neighborhood of
$w$.  Note that it follows from this definition that the dual for
belief $\check{B} \varphi$ is true at $w$ iff
$[w]\cap\semn{\neg\varphi}^\M\notin N(w)$.  The latter says that the
epistemically possible $\lnot\varphi$-worlds do not make up a
neighborhood of $w$.

\subsection{Neighborhood and Probability Model Agreement}

Epistemic neighborhood models describe agent knowledge and belief.
Epistemic probability models can be used for the same purpose along
the lines we have discussed above once we establish a belief threshold
$c\in(0,1)\cap\Rat$.  This gives rise to a natural question: is there
some sense in which these two models for knowledge and belief can be
seen to agree?

\begin{definition}[Model Agreement]
  Let $\M=(W,R,V,N)$ be an epistemic neighborhood model. For a
  threshold $c\in(0,1)\cap\Rat$, to say that a probability measure
  $P:\wp(W)\to[0,1]$ \emph{agrees with $\M$ for threshold $c$} means
  we have the following:
  \begin{itemize}
  \item $P$ satisfies full support (i.e., $P(w)\neq0$ for each
    $w\in W$); and

  \item for each $w\in W$ and $X\subseteq[w]$, we have
    \[
    X\in N(w) \quad\text{iff}\quad
    P_w(X):=P(X|[w])>c\enspace.
    \]
  \end{itemize}
  To say that an epistemic probability model $\M'=(W',R',V',P')$
  \emph{agrees with $\M$ for threshold $c$} means that
  $(W',R',V')=(W,R,V)$ and $P'$ agrees with $\M$ for threshold $c$.
  If the threshold $c$ is not mentioned, it is assumed that
  $c=\frac 12$.
\end{definition}

Agreement for threshold $c$ between an epistemic neighborhood model
and an epistemic probability model makes the translation
$c:\Lang_\KB\to\Lang$ (Definition~\ref{definition:translation})
truth-preserving.

\begin{theorem}[Agreement]
  \label{theorem:agreement}\label{BettingTheorem}
  Fix $c\in(0,1)\cap\Rat$, an epistemic neighborhood model $\M$, and
  an epistemic probability model $\M'$. If $\M$ and $\M'$ agree for
  threshold $c$, then we have for each $\varphi \in \Lang_\KB$ that
  \[
  \M,w\modelsn\varphi \quad\text{iff}\quad
  \M',w\modelsp\varphi^c\enspace.
  \]
\end{theorem}
\begin{proof}
  Induction on the structure of $\varphi\in\Lang_\KB$. The non-modal
  cases are obvious.

  We first consider knowledge formulas. Assume $\M,w \modelsn
  K\psi$.  This means $[w]\subseteq\semn{\psi}^{\M}$. Applying the
  induction hypothesis, this is equivalent to
  $[w]\subseteq\semp{\psi^c}^{\M'}$.  By full support, the latter
  holds if and only if
  \[
  P_w(\semp{\psi^c}^{\M'}) =
  \frac{P(\semp{\psi^c}^{\M'}\cap[w])}{P([w])}=1 \enspace,
  \]
  which is what it means to have $\M',w\modelsp P(\psi^c)=1$.  Since
  $P(\psi^c)=1$ is what is abbreviated by $(K\psi)^c$, the result
  follows.

  Now we move to belief formulas. Assume $\M,w \modelsn B\psi$.
  This means that $[w]\cap\semn{\psi}^{\M}\in N(w)$.  Since $\M'$
  agrees with $\M$, the latter holds iff
  $P_w([w]\cap\semp{\psi^c}^{\M'})>c$.  But this is equivalent
  to $P_w(\semp{\psi^c}^{\M'})>c$, which is what it means to have
  $\M',w \modelsp P(\psi^c)>c$.  Since $P(\psi^c)>c$ is what is
  abbreviated by $(B\psi)^c$, the result follows.
\end{proof}

\subsection{Probability Measures on Epistemic Neighborhood Models}

In this subsection, we take up the question of agreement between
epistemic probability models and epistemic neighborhood models from
the point of view of the latter: given an epistemic neighborhood model
and a threshold $c$, can we find an agreeing epistemic probability
model for this threshold? As we will see, we have a full answer only
for the case $c=\frac 12$.  The case for $c\neq\frac 12$ is open,
though we will have some comments on this in the conclusion of the
paper.

To begin, we adapt an example due to Walley and Fine
\cite{WalleyFine1979:vomacp} to show that not every epistemic
neighborhood model gives rise to an agreeing probability measure.

\begin{theorem}[\cite{WalleyFine1979:vomacp}]
  \label{theorem:walleyfine}
  There exists an epistemic neighborhood model $\M$ that has no
  agreeing probability measure for any threshold $c\in(0,1)\cap\Rat$.
\end{theorem}
\begin{proof}
  We adapt Example~2 from \cite[pp.~344-345]{WalleyFine1979:vomacp} to
  the present setting.  Fix $c\in(0,1)\cap\Rat$.  Let
  $\Prop:=\{a,b,c,d,e,f,g\}$. Define:
  \begin{gather*}
    \mathcal{X} := \{efg, abg, adf, bde, ace, cdg, bcf\}\enspace,
    \\
    \mathcal{Y} := \{abcd, cdef, bceg, acfg, bdfg, abef, adeg \}\enspace.
  \end{gather*}
  Notation: in the above sets, $xyz$ denotes $\{x,y,z\}$, and $wxyz$
  denotes $\{w,x,y,z\}$.  Now define
  \[
  \mathcal{N} := \{ X' \mid \exists X \in \mathcal{X}: X \subseteq X'
  \subseteq\Prop \}\enspace.
  \]
  Let $\M:=(W,R,V,N)$ be defined by $W:=\Prop$, $R:=W\times W$,
  $V(w):=\{w\}$ for each $w\in\Prop$, and $N(w) := \mathcal{N}$ for
  each $w \in W$.  It is straightforward to verify that $\M$ is an
  epistemic neighborhood model and that
  $\mathcal{Y} \cap \mathcal{N} = \emptyset$.

  Toward a contradiction, suppose there exists a probability measure
  $P$ that agrees with $\M$. Since each letter $p\in W$ occurs in
  exactly three of the seven members of $\mathcal{X}$, we have:
  \[
  \sum_{X \in \mathcal{X}} P(X) = \sum_{p\in W}3\cdot P(\{p\})\enspace.
  \]
  Since each letter $p\in W$ occurs in exactly four of the seven
  members of $\mathcal{Y}$, we have:
  \[
  \sum_{Y \in \mathcal{Y}} P(Y) = \sum_{p\in W} 4\cdot P(\{p\})
  > \sum_{X \in \mathcal{X}} P(X)\enspace.
  \]
  On the other hand, since $\mathcal{Y}\cap\mathcal{N}=\emptyset$, no
  member of $\mathcal{Y}$ is a neighborhood of $\M$ and therefore it
  follows by the agreement of $P$ with $\M$ that we have
  $P(Y)\leq c<P(X)$ for each $Y\in\mathcal{Y}$ and $X\in\mathcal{X}$.
  But then
  \[
  \sum_{Y\in\mathcal{Y}}P(Y) < \sum_{X\in\mathcal{X}}P(X)\enspace,
  \]
  and we have reached a contradiction. Conclusion: no such $P$ exists.
\end{proof}

Question: what are the additional restrictions on the neighborhood
function that one must impose in order to guarantee the existence of
an agreeing probability measure for a given threshold
$c\in(0,1)\cap\Rat$?  For $c=\frac 12$, the restrictions are known.
For thresholds $c\neq\frac 12$, the question is open.

The restrictions needed for $c=\frac 12$ were studied first in the
form of a purely probabilistic semantics (i.e., something like
epistemic probability models and not something like our epistemic
neighborhood models).  To our knowledge, Lenzen's
\cite{Lenzen1980:gwuw} is the first complete study of the restrictions
needed in such a purely probabilistic framework over a unary modal
language similar to $\Lang_\KB$. The conditions Lenzen proposed are
targeted to satisfy the conditions of a theorem due to Scott, which is
the key result that gives rise to a probability measure in the
completeness proof for Lenzen's logic.  Here we state the required
restrictions in the language of our epistemic neighborhood models.
Later we will make the link with Lenzen's axiomatic system when we
consider axiomatic theories in the language $\Lang_\KB$ targeted to
our epistemic neighborhood models.

\begin{definition}[Extra Properties for ``Mid-Threshold'' Models]
  \label{definition:extra-properties}
  Let $\M=(W,R,V,N)$ be an epistemic neighborhood model.  For
  $m\in\Int^+$ and sets of worlds $X_1,\dots,X_m$ and $Y_1,\dots,Y_m$,
  we write
  \begin{equation}
    X_1,\dots,X_m\mathbb{I}Y_1,\dots,Y_m
    \label{eq:semantic-lenzen}
  \end{equation}
  to mean that for each $v\in W$, the number of $X_i$'s containing $v$
  is less than or equal to the number of $Y_i$'s containing $v$. This
  is the semantic counterpart of the formula from
  Definition~\ref{definition:segerberg-notation}.  We may write
  $(X_i\mathbb{I}Y_i)_{i=1}^m$ as an abbreviation for
  \eqref{eq:semantic-lenzen}.  Also, we write
  $(X_i\mathbb{E}Y_i)_{i=1}^m$ to mean that both
  $(X_i\mathbb{I}Y_i)_{i=1}^m$ and $(Y_i\mathbb{I}X_i)_{i=1}^m$
  hold, and we allow the notation with $\mathbb{E}$ to be used in a
  form as in \eqref{eq:semantic-lenzen}.  The following is a list of
  properties that $\M$ may satisfy.
  \begin{description}
  \item[(d)] $\forall X \in N(w): [w] - X \notin  N(w)$.

 
  \item[(sc)] $\forall X,Y\subseteq[w]$: if $[w]-X\notin N(w)$
    and $X\subsetneq Y$, then $Y\in N(w)$.

  \item[(scott)] $\forall m\in\Int^+,\forall
    X_1,\dots,X_m,Y_1,\dots,Y_m\subseteq[w]:$
    \[
    \renewcommand{\arraystretch}{1.3}
    \begin{array}{ll}
      \text{if }
      &
      \begin{array}[t]{l}
        X_1,\dots,X_m\mathbb{I}Y_1,\dots,Y_m\quad\text{and}
        \\
        X_1\in N(w)\quad\text{and}
        \\
        \forall i\in\{2,\dots,m\}:
        [w]-X_i\notin N(w) \enspace\text{,}
      \end{array}
      \\
      \text{then }
      &
      \exists j\in\{1,\dots,m\}: Y_j\in N(w)\enspace\text{.}
    \end{array}
    \]
  \end{description}
  To say an epistemic neighborhood model is \emph{mid-threshold} means
  it satisfies (d), (sc), and (scott).  We may drop the word
  ``epistemic'' in referring to mid-threshold epistemic neighborhood
  models.  Pointed versions of mid-threshold neighborhood models are
  defined in the obvious way.
\end{definition}

Property (d) ensures that beliefs are consistent in the sense that the
agent does not believe both $X$ and its complement $[w]-X$.
Property (sc) is a form of ``strong commitment'': if the agent does
not believe the complement $[w]-X$, then she must believe any
strictly weaker $Y$ implied by $X$.  Property (scott) is a version of
the syntactic scheme (Scott) from
Definition~\ref{definition:scott-schemes}.

Let us return to the model $\M$ from the proof of
Theorem~\ref{theorem:KB-probability-incompleteness}.  It is easy to
see that for no $X_i\in\mathcal{X}$ do we have
$W-X_i\in N(a)=\mathcal{N}$.  So, numbering the members of
$\mathcal{X}$ as $X_1,\dots,X_7$ and the members of $\mathcal{Y}$ as
$Y_1,\dots,Y_7$, we see that $\M$ satisfies
\[
(X_i\mathbb{I}Y_i)_{i=1}^7,\enspace X_1\in N(a)\text{, and}\enspace
\forall i\in\{2,\dots,7\}:W-X_i\notin N(a)\enspace,
\]
which is the antecedent of property (scott) from
Definition~\ref{definition:extra-properties}.  However, $\M$ does not
satisfy
\[
\exists j\in\{1,\dots,7\}:Y_j\in N(a)\enspace,
\]
which is the corresponding consequent of the indicated instance of
(scott).  So we see that if we were to restrict ourselves to the class
of epistemic neighborhood models satisfying this property, we would no
longer be able to use $\M$ as a counterexample to the claim that not
every epistemic neighborhood model gives rise to an agreeing
probability measure.  Of course ruling out $\M$ as a counterexample to
this claim does not prove the claim.  However, utilizing (scott) in
conjunction with (d) and (sc), we are able to prove the claim.  This
proof makes crucial use of a theorem due to Scott that is closely 
related to \cite[Theorem 4.1]{Sco64:JMP}. 

In preparation for the statement of Scott's theorem, we recall some
well-known notions from linear algebra.  For a nonempty set $S$, let
$L(S)$ denote the $S$-dimensional real vector space whose vectors
consist of functions $x:S\to\mathbb{R}$ and whose operations of vector
addition and scalar multiplication are defined coordinate-wise: given
vectors $x,y:S\to\mathbb{R}$ and a scalar real $r\in\mathbb{R}$, the
vector $(x+y):S\to\mathbb{R}$ is defined by $(x+y)(s):=x(s)+y(s)$ for
each coordinate $s\in S$ and the vector $(r\cdot x):S\to\mathbb{R}$ is
defined by $(r\cdot x)(s):=r\cdot x(s)$ for each coordinate $s\in S$.
Note that we have just used the usual notational overloading wherein the $+$
or $\cdot$ symbol on one side of an equation refers to the vector
operation, and yet the same symbol on the other side of the same
equation refers to the operation in $\mathbb{R}$. Other common
notational abbreviations such as omission of $\cdot$'s and writing
$-x$ for $(-1)\cdot x$ will be used. To say that a vector
$x:S\to\mathbb{R}$ is \emph{rational} means that all of its
coordinates (i.e., values) are rational numbers.  To say a set
$X\subseteq L(S)$ of vectors is rational means that every vector in
$X$ is rational, and to say that $X$ is \emph{symmetric} means that
$X=-X:=\{-x\mid x\in X\}$.  A \emph{linear functional on $L(S)$} is a
function $f:L(S)\to\mathbb{R}$ satisfying the following property of
\emph{linearity\/}: for each $r_1,r_2\in\mathbb{R}$ and $x,y\in L(S)$,
we have $f(r_1x+r_2y)=r_1\cdot f(x)+r_2\cdot f(y)$.

\begin{theorem}[{\cite[Theorem 1.2]{Sco64:JMP}}]
  \label{theorem:scott}
  Let $S$ be a finite nonempty set and $X$ be a finite, rational,
  symmetric subset of $L(S)$. For each $N\subseteq X$, there exists a
  linear functional $f$ on $L(S)$ that \emph{realizes $N$}, meaning
  \[
  N = \{x\in X\mid f(x)\geq 0\}\enspace,
  \]
  if and only if the following conditions are satisfied:
  \begin{itemize}
  \item for each $x\in X$, we have $x\in N$ or $-x\in N$; and

  \item for each integer $n\geq 0$ and $x_0,\dots,x_n\in N$, we have
    \[
    \sum_{i=0}^n x_i = 0 \quad\Rightarrow\quad -x_0\in N\enspace.
    \]
  \end{itemize}
\end{theorem}

We use this theorem to show that mid-threshold models always give rise
to an agreeing probability measure.  That is, the neighborhood
function of mid-threshold models picks out exactly those neighborhoods
that may be assigned a probability exceeding $\frac 12$.  Many of the
key ideas of the proof of the following result are due to Lenzen
\cite{Lenzen1980:gwuw}. However, the argument we present here has been
rewritten in a streamlined, modern form and in the language of our
epistemic neighborhood models. Despite this difference (and the
necessary work we had to undertake to translate these results into
this modern form), we are happy to credit Professor Lenzen for the
following result.

\begin{theorem}[\cite{Lenzen1980:gwuw}]
  \label{theorem:lenzen}
  Let $\M=(W,R,V,N)$ be a mid-threshold epistemic neighborhood model.
  There exists a probability measure $P:\wp(W)\to[0,1]$ agreeing with
  $\M$ for threshold $\frac 12$; that is,
  \begin{itemize}
  \item $P$ satisfies full support (i.e., $P(w)\neq0$ for each
    $w\in W$); and

  \item for each $w\in W$ and $X\subseteq[w]$, we have
    \[
    \textstyle X\in N(w) \quad\text{iff}\quad
    P_w(X):=P(X|[w])>\frac 12\enspace.
    \]
  \end{itemize}
\end{theorem}
\begin{proof}
  We credit Lenzen \cite{Lenzen1980:gwuw} for this proof, though we
  herein provide an original reformulation of his work within the
  setting of the epistemic neighborhood models introduced in this
  paper.  Proceeding, for $w\in W$, define $S_w:=[w]$.  For each
  $X\subseteq S_w$, define the relative complement $X':=S_w-X$ and let
  $\iota(X):S_w\to\{0,1\}$ be the characteristic function of $X$:
  \[
  \iota(X)(s):=\begin{cases}
    1 & \text{if } s\in X, \\
    0 & \text{otherwise.}
  \end{cases}
  \]
  We consider the following finite subsets of $L(S_w)$:
  \begin{eqnarray*}
    \mathcal{A}_w &:=& 
    \{\iota(X)\mid X\subseteq S_w\}\enspace, 
    \\
    \mathcal{B}_w &:=& 
    \{\iota(X)-\iota(X')\mid X\subseteq S_w 
    \text{ \& } X'\notin N(w)\}\enspace,
    \\
    \mathcal{N}_w &:=&
    \mathcal{A}_w\cup\mathcal{B}_w\enspace,
    \\
    \mathcal{X}_w &:=&
    \mathcal{N}_w\cup(-\mathcal{N}_w)\enspace.
  \end{eqnarray*}
  It is easy to see that $\mathcal{N}_w\subseteq\mathcal{X}_w$
  and that $\mathcal{X}_w$ is a finite, rational, and symmetric
  subset of $L(S_w)$.  We wish to show that $\mathcal{N}_w$
  and $\mathcal{X}_w$ satisfy the conditions of
  Theorem~\ref{theorem:scott}. First, we note that
  $x\in\mathcal{X}_w$ implies $x\in\mathcal{N}_w$ or
  $-x\in\mathcal{N}_w$ by the definition of $\mathcal{X}_w$.

  For the second condition of Theorem~\ref{theorem:scott}, suppose we
  are given an integer $n\geq 0$ such that
  $x_0,\dots,x_n\in\mathcal{N}_w$ and $\sum_{i=0}^n x_i=0$.  We
  wish to show that $-x_0\in\mathcal{N}_w$.  Proceeding, there
  exists an integer $\ell$ satisfying:
  \begin{eqnarray*}
    0\leq i\leq\ell & \text{implies} &
    x_i=\iota(X_i)-\iota(X_i')\in\mathcal{B}_w 
    \enspace,\text{and}
    \\
    \ell<i\leq n & \text{implies} &
    x_i=\iota(X_i)\in\mathcal{A}_w \enspace.
  \end{eqnarray*}
  Toward a contradiction, assume there exists $i>\ell$ with $x_i\neq
  0$.  Then for $x^*:=\sum_{i=\ell+1}^n x_i$, we have $x^*(s)\geq 0$
  for all $s\in S_w$, and there exists $s^*\in S_w$ with
  $x^*(s^*)>0$. Hence
  \[
  \textstyle \sum_{i=0}^\ell x_i=
  \sum_{i=0}^\ell\bigl(\iota(X_i)-\iota(X'_i)\bigr)=-x^*\enspace,
  \]
  where $-x^*(s^*)<0$ and $-x^*(s)\leq 0$ for all $s\in S_w$.  So for
  each $s\in S_w$, the number of the sets in the list
  $X_0',\dots,X_\ell'$ containing $s$ is greater than or equal to the
  number of the sets in the list $X_0,\dots,X_\ell$ containing $s$.
  Further, $s^*$ is a member of strictly more sets in the former list
  than those in the latter.  By renumbering, we may assume that
  $s^*\in X_0'-X_0$.  Then we have
  \[
  X_0\cup\{s^*\},X_1,\dots,X_\ell \mathbb{I}
  X_0',X_1',\dots,X_\ell'\enspace.
  \]
  Since $X_0',\dots,X_\ell'\notin N(w)$, it follows by (scott) that
  $X_0\cup\{s^*\}\notin N(w)$.  But $X_0\subsetneq
  X_0\cup\{s^*\}\notin N(w)$ and $X_0'\notin N(w)$, which violates
  (sc).  Conclusion: $i>\ell$ implies $x_i=0$.  But then we have
  $\sum_{i=0}^n x_i=\sum_{i=0}^\ell x_i$.  Since
  $x_i=\iota(X_i)-\iota(X_i')$ for $i\leq\ell$, it follows that
  $\sum_{i=0}^\ell\iota(X_i)=\sum_{i=0}^\ell\iota(X_i')$. But the
  latter is what it means to have $(X_i\mathbb{E}X_i')_{i=0}^\ell$.
  Since $X_i'\notin N(w)$ for $i\leq\ell$ by the definition of
  $\mathcal{B}_w$, it follows by (scott) that $X_0\notin N(w)$.
  But then $\iota(X'_0)-\iota(X_0)=-x_0\in\mathcal{B}_w \subseteq
  \mathcal{N}_w$, as desired.

  So we may apply Theorem~\ref{theorem:scott}: there exists a linear
  functional $f_w$ on $L(S_w)$ that realizes
  $\mathcal{N}_w$.  That is,
  \[
  \mathcal{N}_w=\{x\in\mathcal{X}_w\mid f_w(x)\geq
  0\}\enspace.
  \]
  Define $g_w:\wp(S_w)\to\mathbb{R}$ by the composition
  $g_w(X):=f_w(\iota(X))$.  This function satisfies a few
  important properties.
  \begin{enumerate}
  \item \label{prop:XinN} $X\in N(w)$ iff $g_w(X)>g_w(X')$.

    Suppose $X\in N(w)$. Then $X'\notin N(w)$ by (d).  Hence
    $\iota(X)-\iota(X')\in\mathcal{B}_w$ and
    $\iota(X')-\iota(X)\notin\mathcal{B}_w$.  Since $S_w\in
    N(w)$ by (n), it follows that $X\neq\emptyset=S_w'$.  But
    then the coordinates of $\iota(X')-\iota(X)$ contain at least one
    $1$ and at least one $-1$.  Since every $x\in\mathcal{A}_w$
    has coordinates that are $1$'s or $0$'s only, it follows that
    $\iota(X')-\iota(X)\notin\mathcal{N}_w$.  As
    $\iota(X)-\iota(X')\in\mathcal{B}_w\subseteq
    \mathcal{N}_w$ and $f_w$ is linear and realizes
    $\mathcal{N}_w$, it follows that $g_w(X)\geq g_w(X')$
    and $g_w(X')\ngeq g_w(X)$.  That is,
    $g_w(X)>g_w(X')$.

    Conversely, suppose $g_w(X)>g_w(X')$. Since $f_w$ is
    linear and realizes $\mathcal{N}_w$, it follows that
    $\iota(X')-\iota(X)\notin\mathcal{N}_w\supseteq
    \mathcal{B}_w$. Applying the definition of
    $\mathcal{B}_w$, we have $X\in N(w)$.

  \item \label{prop:emptyS} $g_w(S_w)>g_w(\emptyset)=0$.

    We have $g_w(\emptyset)=f_w(0)=0$ by the linearity of
    $f_w$.  Since $S_w\in N(w)$ by (n), it follows that
    $g_w(S_w)>g_w(\emptyset)$ by property \ref{prop:XinN}.

  \item \label{prop:zerotoS} If $0\leq g_w(X)\leq
    g_w(S_w)$.
    
    Since $\iota(X)\in\mathcal{A}_w\subseteq \mathcal{N}_w$
    and $f_w$ realizes $\mathcal{N}_w$, we have
    $g_w(X)\geq 0$.  So each $X\subseteq S_w$ satisfies
    $g_w(X)\geq 0$.  From this it follows by the linearity of
    $f_w$ that for each $X\subseteq S_w$, we have
    \[
    \textstyle g_w(X)=\sum_{v\in X}g_w(\{v\}) \leq\sum_{v\in
      S_w}g_w(\{v\}) =g_w(S_w)\enspace.
    \]

  \item \label{prop:additivity} If $X,Y\subseteq S_w$ and $X\cap
    Y=\emptyset$, then $g_w(X\cup Y)=g_w(X)+g_w(Y)$.
    
    By the linearity of $f_w$.

  \item \label{prop:fullsupport} $\emptyset\neq X\subseteq S_w$
    implies $g_w(X)>0$.

    Suppose $\emptyset\neq X\subseteq S_w$. By property
    \ref{prop:emptyS}, it suffices to prove the result for $X\neq
    S_w$. Toward a contradiction, assume $g_w(X)=0$ for
    $\emptyset\subsetneq X\subsetneq S_w$.  By property
    \ref{prop:additivity}, we have
    $g_w(S_w)=g_w(X)+g_w(X')=g_w(X')$.  Since
    $f_w$ is linear and realizes $\mathcal{N}_w$ and
    \[
    \iota(X')-\iota(S_w) = -(\iota(S_w)-\iota(X')) =
    -\iota(X)\in\mathcal{X}_w\enspace,
    \]
    we obtain $-\iota(X)\in\mathcal{N}_w$.  But
    $\emptyset\subsetneq X\subsetneq S_w$ implies that $-\iota(X)$
    has coordinates containing at least one $-1$ and at least one $0$.
    Since members of $\mathcal{A}_w$ have coordinates made up of
    $0$'s and $1$'s, members of $\mathcal{B}_w$ have coordinates
    made up of $-1$'s and $1$'s, and
    $\mathcal{N}_w=\mathcal{A}_w \cup \mathcal{B}_w$, it
    cannot be the case that $-\iota(X)\in\mathcal{N}_w$.
    Contradiction.  Conclusion: $g_w(X)>0$.
  \end{enumerate}

  Now take $v\in[w]$.  Since $N(v)=N(w)$ by (a), it follows that $g_w$
  also realizes $\mathcal{N}_{v}$.  So, letting $[W]$ be the set
  $\{[w]\mid w\in W\}$ of equivalence classes, let $h:[W]\to W$ be a
  choice function that selects for each class $[w]\in[W]$ a
  representative $h([w])\in[w]$. Using a notational overloading that ought
  to be harmless, we define a new function $h_w:\wp([w])\to\mathbb{R}$
  by setting $h_w(X):=g_{h([w])}(X)$.  Obviously, $v\in[w]$ implies
  $h_{v}=h_w$.  Finally, we define $P:\wp(W)\to[0,1]$ by
  \[
  P(X):= \sum_{[w]\in[W]} \frac {h_w(X\cap[w])}
  {h_w([w])}\enspace.
  \]
  Note that by property \ref{prop:emptyS}, the denominator
  $h_w([w])$ is always nonzero.

  We prove that $P$ is a probability measure on $\wp(W)$ satisfying
  full support. First, $P$ satisfies the Kolmogorov axioms over the
  finite algebra $\wp(W)$: we have $P(X)\geq 0$ by property
  \ref{prop:zerotoS}, $P(W)=1$ by property \ref{prop:emptyS} and the
  definition of $P$, and $P(X\cup Y)=P(X)+P(Y)$ for disjoint
  $X$ and $Y$ by property \ref{prop:additivity} and the definition of
  $P$.  Second, full support follows by property
  \ref{prop:fullsupport}.

  Finally, for $X\subseteq[w]$, we have by property \ref{prop:XinN}
  that $X\in N(w)$ iff $h_w(X)>h_w(X')$.  But the latter
  holds iff we have (making use of property \ref{prop:additivity})
  that
  \[
  2\cdot h_w(X)>h_w(X)+h_w(X')=h_w([w])\enspace.
  \]
  By property \ref{prop:emptyS}, the definition of $P$, and the fact
  that $X\subseteq[w]$, the above inequality holds iff
  \[
  P(X)=\frac{h_w(X)}{h_w([w])}>\textstyle \frac
  12\enspace.
  \]
\end{proof}

\begin{corollary}
  \label{corollary:lenzen}
  Let $\M=(W,R,V,N)$ be a mid-threshold epistemic neighborhood model.
  There exists an epistemic probability model $\N=(W,R,V,P)$ that
  agrees with $\M$ for threshold $\frac 12$.
\end{corollary}
\begin{proof} 
  Let $P$ be the measure given by Theorem~\ref{theorem:lenzen}.
\end{proof}

\subsection{Epistemic Neighborhood Models from Probability Measures}
\label{Section:BeliefBet}

In the last subsection, we investigated the question of whether an
epistemic neighborhood model gives rise to an agreeing epistemic
probability model.  In this section, we look at this question the
other way around: given an epistemic probability model and a threshold
$c$, is there an agreeing epistemic neighborhood model?  As we will
see, the answer is always ``yes.''

\begin{definition}
  Given an epistemic probability model $\M = (W,R,V,P)$ and a
  threshold $c\in[\frac 12,1)\cap\Rat$, we define the structure $\M^c
  := (W,R,V,N^c)$ by setting
  \[
  N^c(w) := \{X\subseteq[w]\mid P_w(X)>c\}\enspace.
  \]
\end{definition}

Intuitively, the agent believes a proposition $X$ at world $w$ (i.e.,
$X\in N^c(w)$) if and only if $X$ is epistemically possible (i.e.,
$X\subseteq[w]$) and the probability she assigns to $X$ at world $w$
exceeds the threshold (i.e., $P_w(X)>c$).

\begin{lemma}[Correctness]
  \label{lemma:correctness}
  Fix $c\in(0,1)\cap\Rat$.  If $\M$ is an epistemic probability model,
  then $\M^c$ is an epistemic neighborhood model.  Furthermore,
  $\M^{\frac 12}$ is a mid-threshold neighborhood model.
\end{lemma}
\begin{proof}
  We verify that $N^c$ satisfies the required properties.
  \begin{itemize}
  \item For (kbc), $X\in N^c(w)$ implies $X\subseteq[w]$ by
    definition.

  \item For (kbf), $P_w(\emptyset)=0<c$, so $\emptyset\notin N^c(w)$.

  \item For (n), $P_w([w]) = 1 > c$, so
    $[w] \in N^c(w)$.

  \item For (a), suppose $X \in N^c(w)$ and $v\in [w]$.  Then
    $P_w(X)>c$.  Since $v\in [w]$ implies $[w] = [v]$, we
    have
    \[
    P_w(X) = 
    \frac{P(X\cap[w])}{P([w])} =
    \frac{P(X\cap[v])}{P([v])} =
    P_{v}(X) \enspace.
    \]
    Hence $P_{v}(X)>c$, so $X \in N^c(v)$.

  \item For (kbm), suppose $X \in N^c(w)$.  Then $P_w(X)>c$.
    Hence if $Y$ satisfies $X \subseteq Y \subseteq [w]$, we have
    $P_w(Y)>c$ and so $Y \in N^c(w)$.
  \end{itemize}
  So $\M^c$ is an epistemic neighborhood model.  We now show that
  $\M^{\frac 12}$ satisfies the additional required properties.
  \begin{itemize}
  \item For (d), assume $c\in[\frac 12,1)\cap\Rat$ and $X \in
    N^c(w)$.  Then $P_w(X) > c$, and therefore $P_w([w] -
    X) \leq 1-c\leq c$. Hence $[w] - X \notin N^c(w)$.

  \item For (sc), assume $X':=[\Gamma]-X\notin N^{\frac 12}(w)$
    and $X\subsetneq Y\subseteq[\Gamma]$.  From the first
    assumption, we have $P_w(X') \leq \frac 12$, and therefore
    that $P_w(X)\geq \frac 12$.  Applying the second assumption,
    $P_w(Y) > P_w(X)\geq \frac 12$, and hence $X\in N^{\frac
      12}(w)$.

  \item For (scott), we assume $c\in(0,\frac 12]\cap\Rat$ along with the
    following:
    \begin{eqnarray}
      &&
      (X_i\mathbb{I}Y_i)_{i=1}^m
      \label{eq:prop-l:E2} 
      \\ &&
      X_1\in N^c(w) 
      \label{eq:prop-l:X1} 
      \\ && 
      \forall i\in\{2,\dots,m\}: [\Gamma]-X_i\notin N^c(w) 
      \label{eq:prop-l:Xcs} 
    \end{eqnarray}
    From \eqref{eq:prop-l:E2} it follows that
    \begin{equation}
      P_w (X_1)+\cdots+P_w(X_m)\leq
      P_w(Y_1)+\cdots+ P_w(Y_m)
      \label{eq:sums-eq}
    \end{equation}
    The argument for this is similar to an argument for \eqref{eq:sum}
    in proof of Theorem~\ref{theorem:belief}\eqref{item:B-Len}.  From
    \eqref{eq:prop-l:X1}, we have $P_w(X_1)>c$.  From
    \eqref{eq:prop-l:Xcs}, we have for each $i\in\{2,\dots,m\}$ that
    $P_w([w]-X_i)\leq c$ and therefore that $P_w(X_i)\geq
    1-c\geq c$ since $c\in(0,\frac 12]\cap\Rat$.  Hence the left side
    of \eqref{eq:sums-eq} exceeds $mc$.  Since every summand on the
    right side of the inequality is positive and $mc>0$, it follows
    that at least one member of the right side of \eqref{eq:sums-eq}
    must exceed $c$.  That is, there exists $j\in\{1,\dots,m\}$ such
    that $P_w(Y_j) > c$ and hence $Y_j\in N^c(w)$.  \qedhere
  \end{itemize}
\end{proof}

\begin{theorem}
  Let $c\in(0,1)\cap\Rat$ and $\M=(W,R,V,P)$ be an epistemic
  probability model. The epistemic neighborhood model
  $\M^c=(W,R,V,N^c)$ agrees with $\M$ for threshold $c$.
\end{theorem}
\begin{proof}
  By definition of $N^c$.
\end{proof}

\section{Calculi for Belief as Willingness to Bet}
\label{Section:Calculi}

We now consider an axiomatic link both with epistemic neighborhood
models and with epistemic probability models.  We study two calculi:
the calculus $\KB$ of epistemic neighborhood models, and the calculus
$\KBeq$ of mid-threshold neighborhood models.  Regarding the
probability interpretation, $\KB$ is sound for every threshold but not
complete for any threshold.  $\KBeq$ is both sound and complete for
the probability interpretation with threshold $c=\frac 12$.

$\KBeq$ is our modern reformulation of Lenzen's \cite{Lenzen1980:gwuw}
calculus for the logic of knowledge (i.e., Lenzen's ``acceptance'') as
probabilistic certainty and belief as probability exceeding threshold
$\frac 12$.  Lenzen's intended semantic structures are something like
epistemic probability models.  Our intended semantic structures are
our mid-threshold neighborhood models, though there is a natural link
with epistemic probability models via Theorem~\ref{theorem:lenzen}.
In fact, many of the main ideas of our proof of
Theorem~\ref{theorem:lenzen} are not doubt translations of Lenzen's
ideas into the language of our epistemic neighborhood models. Since we
have rewritten all proofs using our own approach and modern modal
notions, it is difficult to determine whether we have introduced novel
mathematical results on top of Lenzen's existing work, though we
suspect that anything new we may have added along these lines
(excluding of course epistemic neighborhood models themselves and all
related results except Theorem~\ref{theorem:lenzen}) may be slight at
best.  Therefore, we are happy to credit Professor Lenzen for the
probabilistic soundness and completeness of $\KBeq$ and for
Theorem~\ref{theorem:lenzen}.  Nevertheless, we do think that it is
worth our effort to provide this modern reformulation of his results.
In particular, we believe that in using semantic structures more
familiar to the modern modal logician, our modern reformulation of
Lenzen's results will make the mathematical details of Lenzen's work
more accessible to a modern English-language audience.  We also hope
that our use of the modal neighborhood structures will suggest
directions for further study of qualitative probability via tools from
modal logic.

\begin{definition}
  \label{definition:calculi}
  We define the following theories in the language $\Lang_\KB$.
  \begin{itemize}
  \item $\KB$ is defined in Table~\ref{table:KB}.

  \item $\KBeq$ is obtained from $\KB$ by adding (D), (SC), and (Scott)
    from Table~\ref{table:additional-schemes}.

  \item $\KBeqm$ is obtained from $\KBeq$ by omitting (BF) and (KBM).
  \end{itemize}
\end{definition}

We will see later in Theorem~\ref{theorem:KBminus} that $\KBeq$ and
$\KBeqm$ derive the same theorems.

\begin{table}[ht]
  \begin{center}
    \textsc{Axiom Schemes}\\[.4em]
    \renewcommand{\arraystretch}{1.3}
    \begin{tabular}[t]{cl}
      (CL) &
      Schemes of Classical Propositional Logic
      \\
      (KS5) &
      $\mathsf{S5}$ axiom schemes for each $K$
      \\
      (BF) &
      $\lnot B\bot$
      \\
      (N) &
      $B\top$
      \\
      (Ap) &
      $B\varphi\to KB\varphi$
      \\
      (An) &
      $\lnot B\varphi\to K\lnot B\varphi$
      \\
      (KBM) &
      $K(\varphi\to\psi)\to(B\varphi\to B\psi)$
    \end{tabular}
    \renewcommand{\arraystretch}{1.0}
    \\[1em]
    \textsc{Rules}\vspace{-.5em}
    \[
    \begin{array}{c}
      \varphi\to\psi \quad \varphi
      \\\hline
      \psi
    \end{array}\;\text{\footnotesize(MP)}
    \qquad
    \begin{array}{c}
      \varphi
      \\\hline
      K\varphi
    \end{array}\;\text{\footnotesize(MN)}
    \]
  \end{center}
  \caption{The theory $\KB$}
  \label{table:KB}
\end{table}

\begin{table}[ht]
  \begin{center}
    \renewcommand{\arraystretch}{1.3}
    \begin{tabular}[t]{cl}
      (D) &
      $B\varphi\to \check B\varphi$
      \\
      (SC) &
      $\check B\varphi \land 
      \check K(\lnot\varphi\land\psi) \to 
      B(\varphi\lor\psi)$
      \\
      (Scott) &
      $\textstyle [(\varphi_i\mathbb{I}\psi_i)_{i=1}^m
      \land B\varphi_1 \land \bigwedge_{i=2}^m \check B\varphi_i] \to
      \bigvee_{i=1}^m B\psi_i$
    \end{tabular}
  \end{center}
  \caption{Additional axiom schemes for the theory $\KBeq$}
  \label{table:additional-schemes}
\end{table}

\subsection{Results for the Basic Calculus $\KB$}

The following result shows that if we restrict attention to provable
statements whose only modality is single-agent belief $B\varphi$, then
$\KB$ is an extension of the minimal modal logic
$\mathsf{EMN45}+\lnot B\bot=\mathsf{EMN45}+(\text{BF})$ obtained by
adding $\mathsf{S5}$-knowledge and the knowledge-belief connection
principles (Ap), (An), and
(KBM).\footnote{$\mathsf{EMN45}+(\text{BF})$ is the logic of
  single-agent belief (without knowledge) having Schemes (CL)
  (Table~\ref{table:KB}), M
  (Theorem~\ref{theorem:KBgt-derivables}\eqref{derivables:Band-andB}),
  (N) (Table~\ref{table:KB}), 4
  (Theorem~\ref{theorem:KBgt-derivables}\eqref{derivables:pos-belief}),
  5
  (Theorem~\ref{theorem:KBgt-derivables}\eqref{derivables:neg-belief}),
  and (BF) (Table~\ref{table:KB}) along with Rules (MP)
  (Table~\ref{table:KB}) and RE
  (Theorem~\ref{theorem:KBgt-derivables}\eqref{derivables:RE}). This
  is a ``monotonic'' system of modal logic satisfying positive and
  negative belief introspection (4 and 5) and the property (BF) that
  falsehood $\bot$ is not believed. See \cite[Ch.~8]{Chellas:ml} for
  details on naming minimal modal logics.} The modal theory $\KBeq$,
which we will see is equivalent to $\KBeq^-$, is a knowledge-inclusive
extension of $\mathsf{EMND45}+(\text{Scott})$ that adds the additional
connection principle (SC).\footnote{$\mathsf{EMND45}+(\text{Scott})$
  is $\mathsf{EMN45}+(\text{BF})$ minus Scheme (BF) plus Schemes (D)
  and (Scott) from Table~\ref{table:additional-schemes}.}  In
Section~\ref{section:kbeq}, we will show that $\KBeq$ is the modal
logic for probabilistic belief with threshold $c=\frac 12$.

\begin{theorem}[$\KB$ Derivables]
  \label{theorem:KBgt-derivables}
  We have each of the following.
  \begin{enumerate}
  \item $\KB\vdash K\varphi\to B\varphi$.
    \label{derivables:KBC}

    ``Knowledge implies belief.''

  \item $\KB\vdash B(\varphi\land\psi)\to(B\varphi\land B\psi)$.
    \label{derivables:Band-andB}

    This is ``Scheme M'' \cite[Ch.~8]{Chellas:ml}.

  \item $\KB\vdash K\varphi\land B\psi\to B(\varphi\land\psi)$.
    \label{derivables:andB-Band}

    If the antecedent $K\varphi$ were replaced by $B\varphi$, then we
    would obtain ``Scheme C'' \cite[Ch.~8]{Chellas:ml}.  So we do not
    have Scheme C outright but instead a knowledge-weakened version:
    in order to conclude belief of a conjunction from belief of one of
    the conjuncts, the other conjunct must be known (and not merely
    believed, as is required by the stronger, non-$\KB$-provable
    Scheme C).

  \item $\KB\vdash K(\varphi\to\psi)\to(\check B\varphi\to\check B\psi)$.
    \label{derivables:check-M}

    This is the dual version of our (KBM).

    \item $\KB\vdash B\varphi\to BB\varphi$.
    \label{derivables:pos-belief}

    This is ``Scheme 4'' for belief \cite[Ch.~8]{Chellas:ml}.

  \item $\KB\vdash \lnot B\varphi\to B\lnot B\varphi$.
    \label{derivables:neg-belief}

    This is ``Scheme 5'' for belief \cite[Ch.~8]{Chellas:ml}.

  \item $\KB\vdash B\varphi\leftrightarrow KB\varphi$.
    \label{derivables:B-KB}

    This says that belief and knowledge of belief are equivalent.

  \item $\KB\vdash \lnot B\varphi\leftrightarrow K\lnot B\varphi$.
    \label{derivables:nB-KnB}

    This says that non-belief and knowledge of non-belief are equivalent.

  \item $\KB\vdash\varphi$ implies $\KB\vdash B\varphi$.
    \label{derivables:B-nec}

    This is the rule of Modus Ponens (or Modal Necessitation),
    sometimes called ``Rule RN'' \cite[Ch.~8]{Chellas:ml}.

  \item $\KB\vdash\varphi\to\psi$ implies $\KB\vdash B\varphi\to B\psi$.
    \label{derivables:Bimp}

    This is ``Rule RM'' \cite[Ch.~8]{Chellas:ml}.

  \item $\KB\vdash\varphi\to\psi$ implies $\KB\vdash\check
    B\varphi\to\check B\psi$.
    \label{derivables:check-Bimp}
    
    This is the dual version of RM.
    
  \item $\KB\vdash\varphi\leftrightarrow\psi$ implies
    $\KB\vdash B\varphi\leftrightarrow B\psi$.
    \label{derivables:RE}

    This is ``Rule RE'' \cite[Ch.~8]{Chellas:ml}.

  \item $\KB\vdash\varphi\to\bot$ implies $\KB\vdash\lnot
    B\varphi$. \label{derivables:GBF}

    This says that no self-contradictory sentence is believed.  This
    may be viewed as a certain generalization of (BF)
    (Table~\ref{table:KB}).
  \end{enumerate}
\end{theorem}
\begin{proof}
  We reason in $\KB$. For \ref{derivables:KBC}, we have $K\varphi\to
  K(\top\to\varphi)$ by elementary modal reasoning.  But then from
  this, $B\top$ by (N), and $K(\top\to\varphi)\to(B\top\to
  B\varphi)$ by (KBM), it follows by classical reasoning that we have
  $K\varphi\to B\varphi$.

  For \ref{derivables:Band-andB}, we derive
  \begin{equation}
    K((\varphi\land\psi)\to\varphi)\to(B(\varphi\land\psi)\to B\varphi)
    \label{eq:derivables:Band-andB}
  \end{equation}
  by (KBM), and the antecedent of \eqref{eq:derivables:Band-andB} by
  (CL) and (MN).  Therefore, the consequent of
  \eqref{eq:derivables:Band-andB} is derivable by (MN).  By a similar
  argument, $B(\varphi\land\psi)\to B\psi$ is derivable.  By classical
  reasoning, \ref{derivables:Band-andB} is derivable.
  
  For \ref{derivables:andB-Band}, we derive
  \begin{eqnarray}
    &&
    K\varphi\to K(\psi\to(\varphi\land\psi)) \enspace\text{and}
    \label{eq:derivables:andB-Band1}
    \\
    &&
       K(\psi\to(\varphi\land\psi))\to(B\psi\to B(\varphi\land\psi)) \enspace\text{.}
    \label{eq:derivables:andB-Band2}
  \end{eqnarray}
  \eqref{eq:derivables:andB-Band1} follows by $\mathsf{S5}$ reasoning.
  \eqref{eq:derivables:andB-Band2} follows by (KBM).  Applying
  classical reasoning to \eqref{eq:derivables:andB-Band1} and
  \eqref{eq:derivables:andB-Band2}, we obtain
  \[
  K\varphi\to (B\psi\to B(\varphi\land\psi))\enspace,
  \]
  from which \ref{derivables:andB-Band} follows by classical
  reasoning.

  For \ref{derivables:check-M}, we derive
  \begin{eqnarray}
    &&
    K(\varphi\to\psi)\to K(\lnot\psi\to\lnot\varphi) \enspace\text{and}
    \label{eq:derivables:check-M1}
    \\
    &&
    K(\lnot\psi\to\lnot\varphi)\to(B\lnot\psi\to B\lnot\varphi) \enspace\text{.}
    \label{eq:derivables:check-M2}
  \end{eqnarray}
  \eqref{eq:derivables:check-M1} follows by $\mathsf{S5}$ reasoning.
  \eqref{eq:derivables:check-M2} follows by (KBM).  Applying classical
  reasoning to \eqref{eq:derivables:check-M1} and
  \eqref{eq:derivables:check-M2}, we obtain
  \[
  K(\varphi\to\psi)\to(B\lnot\psi\to B\lnot\varphi)\enspace,
  \]
  from which \ref{derivables:check-M} follows by classical reasoning
  (just contrapose the consequent).

  \ref{derivables:pos-belief} follows by (Ap) and
  \ref{derivables:KBC}.  \ref{derivables:neg-belief} follows by (An)
  and \ref{derivables:KBC}.  \ref{derivables:B-KB} follows by by (Ap)
  for the right-to-left and (KS5) for the left-to-right.
  \ref{derivables:nB-KnB} follows by (An) for the right-to-left and
  (KS5) for the left-to-right.  \ref{derivables:B-nec} follows by (MN)
  and \ref{derivables:KBC}.  \ref{derivables:Bimp} follows by (MN) and
  (KBM).  \ref{derivables:check-Bimp} follows by contraposition, (MN),
  (KBM), and contraposition.  \ref{derivables:RE} follows from
  \ref{derivables:Bimp} by classical reasoning.

  For \ref{derivables:GBF}, we have
  \begin{equation}
    K(\varphi\to\bot)\to(B\varphi\to B\bot)
    \label{eq:GBF}
  \end{equation}
  by (KBM).  Therefore, if $\varphi\to\bot$ is provable, it follows by
  (MN) that the antecedent of \eqref{eq:GBF} is as well.  By (MP), the
  consequent $B\varphi\to B\bot$ is provable.  Applying (BF) and
  classical reasoning, it follows by contraposition that $\lnot
  B\varphi$ is provable.
\end{proof}

\begin{theorem}[$\KB$ Neighborhood Soundness and Completeness]
  \label{theorem:KB-neighborhood-soundness}\label{theorem:KB-neighborhood-completeness}
  $\KB$ is sound and complete with respect to the class $\mathcal{C}$
  of epistemic neighborhood models:
  \[
  \forall\varphi\in\Lang_\KB:\quad\KB\vdash\varphi
  \quad\Leftrightarrow\quad
  \mathcal{C}\modelsn\varphi
  \enspace.
  \]
\end{theorem}
\begin{proof}
  By induction on the length of derivation.  We first
  verify soundness of the axioms.
  \begin{itemize}
  \item Validity of (CL) immediate. Validity of (KS5) follows because
    the $R$'s are equivalence relations \cite{BlaRijVen:ml}.

  \item Scheme (BF) is valid: $\modelsn\lnot B\bot$.

    $\semn{\bot}=\emptyset\notin N(w)$ by (kbf).  Hence
    $\M,w\not\modelsn B\bot$.

  \item Scheme (N) is valid: $\modelsn B\top$.

    $\semn{\top}\cap[w]=[w]\in N(w)$ by (n).  Hence
    $\M,w\modelsn B\top$.

  \item Scheme (Ap) is valid: $\modelsn B\varphi\to K B\varphi$.

    Suppose $\M,w\modelsn B\varphi$. Then $[w]\cap\semn{\varphi}\in
    N(w)$.  Take $v\in[w]$.  We have $[v]=[w]$ because $R$
    is an equivalence relation, and we have $N(v)=N(w)$ by (a).
    Hence $[v]\cap\semn{\varphi}\in N(v)$; that is, $\M,v\modelsn
    B\varphi$.  Since $v\in [w]$ was chosen arbitrarily, we have
    shown that $[w]\subseteq\semn{B\varphi}$.  Hence $\M,w\modelsn
    K B\varphi$.

  \item Scheme (An) is valid: $\modelsn \lnot B\varphi\to K\lnot
    B\varphi$.

    Replace $B\varphi$ by $\lnot B\varphi$ and $\in$ by $\notin$ in the
    argument for the previous item.

  \item Scheme (KBM) is valid: $\modelsn
    K(\varphi\to\psi)\to(B\varphi\to B\psi)$.

    Suppose $\M,w\modelsn K(\varphi\to\psi)$ and $\M,w\modelsn
    B\varphi$.  This means $[w]\subseteq\semn{\varphi\to\psi}$ and
    $[w]\cap\semn{\varphi}\in N(w)$. But then
    \[
    [w]\cap\semn{\varphi}\subseteq
    [w]\cap\semn{\varphi}\cap\semn{\varphi\to\psi}\subseteq
    [w]\cap\semn{\psi}\enspace.
    \]
    Hence $[w]\cap\semn{\psi}\in N(w)$ by (kbm).  That is,
    $\M,w\modelsn B\psi$.
  \end{itemize}
  That validity is closed under applications of the rules MP and MN
  follows by the standard arguments \cite{BlaRijVen:ml}.  This
  completes the proof of soundness.

  Before we prove completeness, we first prove an important result
  that we will use tacitly throughout the completeness proof
  proper. Let $M$ be the set of all $\Lang_\KB$-formulas having one of
  the forms $K\varphi$, $\lnot K\varphi$, $B\varphi$, or
  $\lnot B\varphi$.  We prove the following \emph{Modal-Assumption
    Deduction Theorem\/}: for each finite $F\subseteq M$, we have
  \begin{center}
    $F\vdash_\KB\varphi$ \quad{}iff\quad $\vdash_\KB(\bigwedge
    F)\to\varphi$\enspace.
  \end{center}
  The right-to-left direction straightforward. The proof of the
  left-to-right direction is by induction on the length of derivation.
  All cases are standard except for the induction step in which (MN)
  is applied, so we focus on this case.  Suppose
  $F\vdash_\KB K\varphi$ is derived by (MP) from $\varphi$ such that
  $F\vdash_\KB\varphi$. By the induction hypothesis, we have
  $\vdash_\KB (\bigwedge_{\chi\in F}\chi)\to\varphi$.  By (MN) and
  $\mathsf{K}$ reasoning, we have
  $\vdash_\KB (\bigwedge_{\chi\in F}K\chi)\to K\varphi$.  However, it
  also follows by $\mathsf{S5}$ reasoning (using schemes $\mathsf{4}$
  and $\mathsf{5}$), scheme (Ap), scheme (An), and the fact that
  $F\subseteq M$ that we have $\vdash_\KB\chi\to K\chi$ for each
  $\chi\in F$.  Hence
  $\vdash_\KB(\bigwedge_{\chi\in F}\chi)\to(\bigwedge_{\chi\in
    F}K\chi)$,
  where $(\bigwedge_{\chi\in F}\chi)=\bigwedge F$.  Conclusion:
  $\vdash_\KB(\bigwedge F)\to K\varphi$.

  To prove completeness, it suffices to show that
  $\KB\nvdash\lnot\theta$ implies $\theta$ is satisfiable at a pointed
  epistemic neighborhood model. For two sets $F$ and $F'$ of
  $\Lang_\KB$-formulas, to say that $F$ is \emph{maxcons in $F'$}
  means that $F\subseteq F'$, the set $F$ is $\KB$-consistent (i.e.,
  for no finite $G\subseteq F$ do we have
  $\vdash_\KB (\bigwedge G)\to\bot$), and adding any formula
  $\psi\in F'$ not already in $F$ will produce a $\KB$-inconsistent
  (i.e., not $\KB$-consistent) set. 

  For a set $F$ of $\Lang_\KB$-formulas, we define the
  \emph{single-negation closure} $\pm F$ of $F$ and the \emph{modal
    closure} $\MCl(F)$ of $F$ to be the sets
  \begin{eqnarray*}
    \pm F &:=& 
               F\cup \{\lnot\varphi\mid\varphi\in F\} \cup
               \{\bot,\top\}
               \enspace,
    \\
    \MCl(F) &:=& 
                 F\cup
                 \{X\varphi\mid \varphi\in F\text{
                 and } X\in\{K,\lnot K,B,\lnot B\}\}\enspace.
  \end{eqnarray*}
  In particular, $\MCl(F)$ is obtained from $F$ by adding for each
  formula $\varphi\in F$ the additional formulas $K\varphi$,
  $\lnot K\varphi$, $B\varphi$, and $\lnot B\varphi$.  We say the each
  of the latter four formulas is a \emph{modalization} of $\varphi$.

  Let $S$ be the set of subformulas of $\theta$, including $\theta$
  itself.  Let $C_0$ be the Boolean closure of $S$; that is, $C_0$ is
  the smallest extension of $S$ that contains the propositional
  constants $\top$ (truth) and $\bot$ (falsehood), or their
  abbreviations in $\Lang_\KB$ if they are not primitive, and is
  closed under the Boolean connectives (e.g., negation, conjunction,
  implication, and disjunction) definable in the language.  Finally,
  define $C:=\MCl(C_0)$.

  We define the structure $\M=(W,R,V,N)$ as follows:
  \begin{eqnarray*}
    W &:=& 
    \{w\subseteq C \mid w \text{ is maxcons in $C$}\}, 
    \\{}
    [\varphi] &:=& \{w\in W\mid \varphi\in w\}
    \text{ for } \varphi\in C,
    \\
    R &:=& \{(w,v)\in W^2\mid w\cap M=v\cap M\},
    \\{}
    V(w) &:=& \Prop\cap w,
    \\
    N(w) &:=&
    \{ X\subseteq[w] \mid 
    \exists\varphi\in C:(X=[\varphi]\cap[w]
    \text{ and }
    w\cap M\vdash_\KB B\varphi)
    \}.
  \end{eqnarray*}

  We make use (often tacitly) of the following \emph{In-class Identity
    Lemma\/}: for each $u,v\in W$, if $[u]=[v]$ and
  $u\cap\pm S=v\cap\pm S$, then $u=v$.  So suppose $[u]=[v]$ and
  $u\cap S'=v\cap S'$.  Given $\varphi\in u$, we wish to show that
  $\varphi\in u$.  There are two cases to consider.
  \begin{itemize}
  \item Case: $\varphi\in u\cap C_0$.

    $\varphi$ is a Boolean combination of members of $S$ and is
    therefore $\KB$-provably equivalent to a formula $\varphi'$ that
    is a disjunction of conjunctions of maxcons subsets of $\pm S$.
    It follows by the maximal $\KB$-consistency of $u$ that
    $\varphi'\in u$ and hence one of the disjuncts $\varphi''$ of
    $\varphi'$ is a member of $u$.  Applying the maimal
    $\KB$-consistency of $u$, it follows from $\varphi''\in u$ that
    every conjunct of $\varphi''$ is a member of $u$.  But each
    conjunct of $\varphi''$ is a member of $\pm S\subseteq S'$ and
    hence each conjunct of $\varphi''$ is a member of $v$ by our
    assumption $u\cap S'=v\cap S'$.  Applying the maximal
    $\KB$-consistency of $v$, it follows that $\varphi''\in v$, hence
    $\varphi'\in v$, and hence $\varphi\in v$.

  \item Case: $\varphi\in u\cap(C-C_0)$.

    $\varphi$ is a modalization $X\psi$ of a Boolean combination of
    members of $S$. But $X\psi\in M$ and our assumption $[u]=[v]$
    implies $u\cap M=v\cap M$.  So $X\psi\in v$.
  \end{itemize}
  The converse is proved similarly.

  The In-class Identity Lemma gives rise to the following
  \emph{Identity Lemma\/}: for each $u,v\in W$, if
  $u\cap\MCl(\pm S)=v\cap\MCl(\pm S)$, then $u=v$. Indeed, suppose
  $u\cap\MCl(\pm S)=v\cap\MCl(\pm S)$.  If $[u]=[v]$, then it follows
  from $\pm S\subseteq\MCl(\pm S)$ that we have
  $u\cap\pm S=v\cap\pm S$, and therefore $u=v$ by the In-class
  Identity Lemma. So it suffices to prove that $[u]=[v]$; that is, we
  prove that $u\cap M=v\cap M$.  Proceeding, take
  $X\varphi\in u\cap M$.  If $X\varphi\in C_0$, then we have
  $X\varphi\in v\cap M$ by the argument in the first case of the
  In-class Identity Lemma.  So suppose $X\varphi\in C-C_0$ so that
  $X\varphi$ is a modalization of $\varphi\in C_0$. We have
  $\vdash_\KB\varphi\leftrightarrow\varphi'$, where $\varphi'$ is a
  disjunction of conjunctions of maxcons subsets of $\pm S$.  Applying
  (MN) and $\mathsf{K}$ reasoning, we obtain
  \begin{equation}
    \vdash_\KB K\varphi\leftrightarrow K\varphi' \quad\text{and}\quad
    \vdash_\KB \lnot K\varphi\leftrightarrow \lnot K\varphi'\enspace.
    \label{eq1:KB-completeness}
  \end{equation}
  Applying (KBM) and classical reasoning to
  \eqref{eq1:KB-completeness}, it follows that
  \begin{equation}
    \vdash_\KB B\varphi\leftrightarrow B\varphi' \quad\text{and}\quad
    \vdash_\KB \lnot B\varphi\leftrightarrow \lnot B\varphi'\enspace.
    \label{eq2:KB-completeness}
  \end{equation}
  Since $X\varphi\in u$, we have by \eqref{eq1:KB-completeness},
  \eqref{eq2:KB-completeness}, and the maximal $\KB$-consistency of
  $u$ that $X\varphi'\in u\cap\MCl(\pm S)$.  Since
  $u\cap\MCl(\pm S)=v\cap\MCl(\pm S)$, it follows that
  $X\varphi'\in v$, and hence $X\varphi\in v$ by the maximal
  $\KB$-consistency of $v$.  The converse is proved similarly.

  We may make use (often tacitly) of the following \emph{Definability
    Lemma\/}: for each $w\in W$ and each $X\subseteq[w]$, defining
  \[
  \textstyle X^d :=\bigvee_{v\in X}\bigwedge(v\cap\pm S)\enspace,
  \]
  it follows that $X^d\in C_0\subseteq C$ and $[X^d]\cap[w]=X$. For
  the proof, first note that $X^d\in C_0$ because $C_0$ is closed
  under Boolean operations and $\pm S\subseteq C_0\subseteq C$.  So
  assume $u\in[X^d]\cap[w]$, which implies $X^d\in u$ and $[u]=[w]$.
  Since $u$ is maxcons in $C\supseteq\pm S$, we have by the above
  definition of $X^d$ as a disjunction over $v\in X$ that there exists
  $v\in X$ such that $\bigwedge(v\cap\pm S)\in u$ and hence
  $v\cap\pm S\subseteq u$.  Since $u$ is maxcons in $C$ and hence
  maxcons in $\pm S$ and since $\pm S$ is closed under the operation
  ${\sim}:\Lang_\KB\to\Lang_\KB$ defined by
  \[
  {\sim}\varphi :=
  \begin{cases}
    \phantom{\lnot}\psi & \text{if }\varphi=\lnot\psi\\
    \lnot\varphi & \text{otherwise,} \\
  \end{cases}
  \]
  it follows that $u\cap\pm S=v\cap\pm S$.  So since $[u]=[v]$ and
  $u\cap\pm S=v\cap\pm S$, we have $u=v\in X$ by the Identity Lemma.
  Conversely, suppose $u\in X\subseteq[w]$.  By the definition of
  $X^d$, we have $\KB\vdash\bigwedge(u\cap\pm S)\to X^d$ and therefore
  $X^d\in u$ because $u$ is maxcons in $C$ and
  $\bigwedge(u\cap\pm S)\in u$. Hence $u\in[X^d]\cap[w]$ because
  $u\in X\subseteq[w]$.

  Our definitions above specify the structure $\M=(W,R,V,N)$.  $W$ is
  nonempty because $\theta$ is consistent and so may be extended to a
  maxcons $w_\theta\in W$.  Since $\MCl(\pm S)$ is finite, it follows
  by the Identity Lemma that $W$ is finite.  Further, $R$ is an
  equivalence relation. So to conclude that $M$ is an epistemic
  neighborhood model, all that remains is for us to show that $N$
  satisfies the neighborhood function properties.
  \begin{description}
  \item[(kbc)] $X\in N(w)$ implies $X\subseteq[w]$.

    By definition.

  \item[(bf)] $\emptyset\notin N(w)$.

    Choose $\varphi\in C$ satisfying $[\varphi]\cap[w]=\emptyset$.  It
    follows that $w\cap M\vdash_\KB\varphi\to\bot$, since otherwise we
    could extend $(w\cap M)\cup\{\varphi\}$ to some
    $v\in[\varphi]\cap[w]$, which would contradict
    $[\varphi]\cap[w]=\emptyset$.  So by (MN), we have
    $w\cap M\vdash_\KB K(\varphi\to\bot)$ and hence
    $w\cap M\vdash_\KB B\varphi\to B\bot$ by (KBM).  Since
    $w\cap M\vdash\lnot B\bot$ by (BF), it follows that
    $w\cap M\vdash_\KB \lnot B\varphi$.  So we have shown that
    $w\cap M\vdash_\KB \lnot B\varphi$ for each $\varphi\in C$
    satisfying $[\varphi]\cap[w]=\emptyset$.  $\KB$ is consistent
    (just apply soundness to any epistemic neighborhood model), and
    therefore we have $w\cap M\nvdash_\KB B\varphi$ for each
    $\varphi\in C$ satisfying
    $[\varphi]\cap[w]=\emptyset$. Conclusion: $\emptyset\notin N(w)$.

  \item[(n)] $[w]\in N(w)$.

    $w\cap M\vdash_\KB B\top$ by (N).  Hence
    $[\top]\cap[w]=[w]\in N(w)$.

  \item[(a)] $v\in[w]$ implies $N(v)=N(w)$.

    $v\in[w]$ implies $[v]=[w]$ and $v\cap M=w\cap M$.  Therefore for
    each $X\subseteq[v]=[w]$, we have $\varphi\in C$ satisfying
    $X=[\varphi]\cap[v]$ and $v\cap M\vdash_\KB B\varphi$ iff
    $X=[\varphi]\cap[w]$ and $w\cap M\vdash_\KB B\varphi$.  Hence
    $N(v)=N(w)$.

  \item[(kbm)] If $X\subseteq Y\subseteq[w]$ and $X\in N(w)$, then
    $Y\in N(w)$.

    Suppose $X\subseteq Y\subseteq[w]$ and $X\in N(w)$. Then there is
    $\varphi\in C$ satisfying $X=[\varphi]\cap[w]$ and
    $w\cap M\vdash_\KB B\varphi$.  Since $X\subseteq Y$, it follows
    that $[\varphi]\cap[w]\subseteq[Y^d]\cap[w]$. From this we obtain
    that $w\cap M\vdash_\KB\varphi\to Y^d$, since otherwise we could
    extend $(w\cap M)\cup\{\varphi,\lnot Y^d\}$ to some
    $v\in[\varphi]\cap[\lnot Y^d]\cap[w]$, which would contradict
    $[\varphi]\cap[w]\subseteq[Y^d]\cap[w]$.  Hence
    $w\cap M\vdash_\KB K(\varphi\to Y^d)$ by (MN) and so
    $w\cap M\vdash_\KB B\varphi\to B Y^d$ by (KBM).  Since
    $w\cap M\vdash_\KB B\varphi$, we have $w\cap M\vdash_\KB BY^d$.
    Hence $Y\in N(w)$.
  \end{description}
  So $M$ is indeed and epistemic neighborhood model. To complete our
  overall argument, it suffices to prove the \emph{Truth Lemma\/}: for
  each $\varphi\in C$ and $w\in W$, we have $\varphi\in w$ iff
  $\M,w\modelsn\varphi$. The argument is by induction on the
  construction of $\varphi\in C$.  Boolean cases are straightforward,
  so we restrict our attention to the modal cases: formulas $B\varphi$
  and $K\varphi$ in $C$.  Note that by the definition of $C$ as the
  Boolean closure of the set $S$ of subformulas of $\theta$, either of
  $B\varphi\in C$ or $K\varphi\in C$ implies $\varphi\in C$.

  Suppose $B\varphi\in w$. Then $w\cap M\vdash_\KB B\varphi$ and hence
  $[\varphi]\cap[w]\in N(w)$ by the definition of $N$ and the fact
  that $\varphi\in C$. Applying the induction hypothesis,
  $[\varphi]=\semn{\varphi}^\M$, so
  $\semn{\varphi}^\M\cap[w]\in N(w)$.  But this is what it means to
  have $\M,w\modelsn B\varphi$.

  Conversely, assume $\M,w\modelsn B\varphi$ for $B\varphi\in C$. This
  means $\semn{\varphi}^\M\cap[w]\in N(w)$.  By the induction
  hypothesis and the fact that $\varphi\in C$, we have
  $[\varphi]=\semn{\varphi}^\M$, so $[\varphi]\cap[w]\in N(w)$.  By
  the definition of $N$, there exists $\psi\in C$ such that
  $w\cap M\vdash_\KB B\psi$ and $[\varphi]\cap[w]=[\psi]\cap[w]$.  But
  then $w\cap M\vdash_\KB\psi\to\varphi$, for otherwise we could
  extend $(w\cap M)\cup\{\psi,\lnot\varphi\}$ to some $v\in[w]$ such
  that $v\in[\psi]\cap[w]$ and $v\notin[\varphi]\cap[w]$,
  contradicting $[\varphi]\cap[w]=[\psi]\cap[w]$.  Applying (MN), we
  have $w\cap M\vdash_\KB K(\psi\to\varphi)$ and hence
  $w\cap M\vdash_\KB B\psi\to B\varphi$ by (KBM).  Since
  $w\cap M\vdash_\KB B\psi$, it follows that
  $w\cap M\vdash_\KB B\varphi$.  And since $w$ is maxcons in $C$, we
  conclude that $B\varphi\in w$.

  Now suppose $K\varphi\in w$.  Then for each $v\in [w]$, we have
  that $K\varphi\in v$ and therefore $\varphi\in v$ by $\mathsf{S5}$
  reasoning (using scheme $\mathsf{T}$) and the fact that $v$ is
  maxcons in $C$.  But then we have shown that
  $[w]\subseteq[\varphi]$.  Since $\varphi\in C$, it follows by the
  induction hypothesis that $[w]\subseteq\semn{\varphi}^\M$, which
  is what it means to have $\M,w\modelsn K\varphi$.

  Conversely, assume $\M,w\modelsn K\varphi$ for $K\varphi\in C$. It
  follows that $[w]\subseteq\semn{\varphi}^\M$.  By the induction
  hypothesis, $[w]\subseteq[\varphi]$.  But then
  $w\cap M\vdash_\KB\varphi$, for otherwise we could extend
  $(w\cap M)\cup\{\lnot\varphi\}$ to some $v\in[w]$ satisfying
  $v\notin[\varphi]$, contradicting $[w]\subseteq[\varphi]$.  By (MN),
  we have $w\cap M\vdash_\KB K\varphi$.  Since $K\varphi\in C$ and $w$
  is maxcons in $C$, it follows that $K\varphi\in w$.
\end{proof}

Since $\KB$ is sound and complete with respect to the class of
epistemic neighborhood models, we would expect that in light of
Theorem~\ref{theorem:walleyfine} that $\KB$ is at most sound for the
probability interpretation.

\begin{theorem}[$\KB$ Probability Soundness]
  $\KB$ is sound for any threshold $c\in(0,1)\cap\Rat$ with respect to
  the class of epistemic probability models:
  \[
  \textstyle \forall c\in(0,1)
  \cap\Rat,\forall\varphi\in\Lang_\KB:\quad
  \KB\vdash\varphi
   \quad\Rightarrow\quad
  {}\modelsp\varphi^c \enspace.
  \]
\end{theorem}
\begin{proof}
  Theorems~\ref{theorem:knowledge} and \ref{theorem:belief}.
\end{proof}

\begin{theorem}[$\KB$ Probability Incompleteness]
  \label{theorem:KB-probability-incompleteness}
  $\KB$ is incomplete for all thresholds $c\in(0,1)\cap\Rat$ with
  respect to the class of epistemic probability models:
  \[
  \textstyle 
  \exists\varphi\in\Lang_\KB,
  \forall c\in(0,1)
  \cap\Rat:\quad
  {}\modelsp\varphi^c
  \quad\text{and}\quad
  \KB\nvdash\varphi\enspace.
  \]
\end{theorem}
\begin{proof}
  Take $\M$ as in the proof of Theorem~\ref{theorem:walleyfine}.  Let
  $\sigma$ be the modal formula describing $(\M,a)$: informally (and
  easily formalizable),
  \[
  \textstyle \sigma \quad:=\quad a\bar{b}\cdots\bar{g}\land
  KW\land(\bigwedge_{Z\in N(a)} BZ)\land
  (\bigwedge_{Z'\in\wp(W)-N(a)}\lnot BZ')\enspace.
  \]
  We have $\M,w\modelsn\sigma$ so that $\not\modelsn\lnot\sigma$ and
  therefore $\KB\nvdash\lnot\sigma$ by
  Theorem~\ref{theorem:KB-neighborhood-completeness}.  By the proof of
  Theorem~\ref{theorem:walleyfine}, there is no probability measure
  agreeing with $\M$ for any threshold.  Hence
  $\modelsp\lnot\sigma^c$. So $\varphi:=\lnot\sigma$ gives us the desired
  formula.
\end{proof}

\subsection{Results for the Mid-Threshold Calculus $\KBeq$}
\label{section:kbeq}

We first show that the $\KB$ schemes (BF) and (KBM) are redundant in
the theory $\KBeq$.

\begin{theorem}
  \label{theorem:KBminus}
  $\KBeqm$ and $\KBeq$ derive the same theorems:
  \[
  \forall\varphi\in\Lang_\KB:\quad
  \KBeqm\vdash\varphi \quad\Leftrightarrow\quad \KBeq\vdash\varphi\enspace.
  \]
\end{theorem}
\begin{proof}
  It suffices to prove that the schemes (BF) and (KBM) are derivable
  in $\KBeqm$. For (KBM), we have by
  Definition~\ref{definition:segerberg-notation} that the formula
  $\varphi\mathbb{I}\psi$ is just
  \begin{equation}
    K\bigl(\;
    \underbrace{(\lnot\varphi\land\lnot\psi) \lor (\lnot\varphi\land\psi)}_{F_0} \lor
    \underbrace{(\varphi\land\psi)}_{F_1}
    \;\bigr)\enspace,
    \label{eq:KaPhitoPsi}
  \end{equation}
  where we have explicitly indicated the subformulas $F_0$ and $F_1$
  used in the notation of
  Definition~\ref{definition:segerberg-notation}.  Semantically,
  \eqref{eq:KaPhitoPsi} says that in each of the agent's accessible
  worlds, $\psi$ is true whenever $\varphi$ is true.  Now reasoning
  within $\KBeqm$, it follows that $K(\varphi\to\psi)$ is provably
  equivalent to $\varphi\mathbb{I}\psi$.  But then from $K(\varphi\to\psi)$
  and $B\varphi$, we may derive $\varphi\mathbb{I}\psi$ and $B\varphi$, from
  which we may derive $B\psi$ by (Scott). Hence (KBM) is derivable.

  We now consider (BF).  The formula $\bot\to\lnot\top$ is a classical
  tautology and hence $K(\bot\to\lnot\top)$ follows by (MN).  Hence
  by an instance of (KBM), which can be defined away in terms of
  axioms other than (BF) as above, it follows that $B\bot\to
  B\lnot\top$ and therefore that $\lnot B\lnot\top\to\lnot
  B\bot$.  Also by (N), (D), and (MP), we may derive $\lnot
  B\lnot\top$.  That is, (BF) is derivable.
\end{proof}

\begin{theorem}[$\KBeq$ Neighborhood Soundness and Completeness]
  \label{theorem:KBeq}
  $\KBeq$ is sound and complete with respect to the class $\Ceq$ of
  mid-threshold neighborhood models:
  \[
  \forall\varphi\in\Lang_\KB:\quad
  \KBeq\vdash\varphi \quad\Leftrightarrow\quad \Ceq\modelsn\varphi \enspace.
  \]
\end{theorem}
\begin{proof}
  Soundness is by induction on the length of derivation.  Most cases
  are as in the proof of
  Theorem~\ref{theorem:KB-neighborhood-completeness}.  We only need
  consider the remaining axiom schemes.
  \begin{itemize}
  \item Scheme (D) is valid: $\modelsn B\varphi\to \check B\varphi$.

    Suppose $\M,w\modelsn B\varphi$. This means
    $[w]\cap\semn\varphi\in N(w)$.  By (d), 
    \[
    [w]\cap\semn{\lnot\psi}=[w]-\semn\varphi =
    [w]-([w]\cap\semn\varphi)\notin N(w)\enspace.
    \]
    But this is what it means to have $\M,w\modelsn\check B\varphi$.

  \item Scheme (SC) is valid: $\modelsn \check B\varphi \land \check
    K(\lnot\varphi\land\psi) \to B(\varphi\lor\psi)$.

    Suppose $\M,w\modelsn\check B\varphi$ and $\M,w\modelsn \check
    K(\lnot\varphi\land\psi)$.  It follows that
    \[
    [w]-([w]\cap\semn{\varphi})=[w]\cap\semn{\lnot\varphi}\notin
    N(w)
    \]
    and that there exists $v\in[w]$ satisfying
    $\M,v\models\lnot\varphi\land\psi$.  But then
    $[w]\cap\semn{\varphi\lor\psi}\supsetneq[w]\cap\semn{\varphi}$ and
    therefore $[w]\cap\semn{\varphi\lor\psi}\in N(w)$ by (sc). Hence
    $\M,w\models B(\varphi\lor\psi)$.

  \item Scheme (Scott) is valid:
    \[
    \modelsn \textstyle [(\varphi_i\mathbb{I}\psi_i)_{i=1}^m
    \land B\varphi_1 \land \bigwedge_{i=2}^m \check B\varphi_i] \to
    \bigvee_{i=1}^m B\psi_i\enspace.
    \]

    Suppose $(\M,w)$ satisfies the antecedent of scheme (Scott).  It
    follows that each $v\in[w]$ satisfies at least as many
    $\varphi_i$'s as $\psi_i$'s, that $[w]\cap\semn{\psi_1}\in N(w)$,
    and that $[w]-\semn{\varphi_k}\notin N(w)$ for each
    $k\in\{2,\dots,m\}$.  Hence
    \[
    [w]\cap\semn{\varphi_1},\dots,[w]\cap\semn{\varphi_m}\mathbb{I}
    [w]\cap\semn{\psi_1},\dots, [w]\cap\semn{\psi_m}\enspace,
    \]
    from which it follows by (scott) that $[w]\cap\semn{\psi_j}\in
    N(w)$ for some $j\in\{1,\dots,m\}$.  Hence $\M,w\modelsn
    B\psi_j$, and thus $\M,w\modelsn\bigvee_{i=1}^m B\psi_i$.
  \end{itemize}
  Soundness has been proved.  

  For completeness, it suffices to show that the model $\M$ defined as
  in the proof of
  Theorem~\ref{theorem:KB-neighborhood-completeness}---except that now
  derivability is always taken with respect to $\KBeq$---is a
  mid-threshold neighborhood model; the rest of the argument is as in
  that proof, \emph{mutatis mutandis}.  Most of the properties of $\M$
  are shown in that proof.  What remains is for us to show that $\M$
  also satisfies (d), (sc), and (scott).
  \begin{description}
  \item[(d)] $X\in N(w)$ implies $X'\notin N(w)$, where
    $X':=[w]-X$.

    Suppose $X\in N(w)$.  Then we have $\varphi\in C$ such that
    $X=[\varphi]\cap[w]$ and $w\cap M\vdash_\KBeq B\varphi$.  By (D),
    it follows that $w\cap M\vdash_\KBeq\check B\varphi$.  Choose
    $\psi\in C$ satisfying $X'=[\psi]\cap[w]$.  We have
    $w\cap M\vdash_\KBeq\psi\to\lnot\varphi$, since otherwise we could
    extend $(w\cap M)\cup\{\psi,\varphi\}$ to a $v\in[w]$ such that
    $v\in[\psi]\cap[w]=X'$ and $v\in[\varphi]\cap[w]=X$, contradicting
    $X'\cap X=\emptyset$.  By (MP), we have
    $w\cap M\vdash_\KBeq K(\psi\to\lnot\varphi)$ and therefore
    $w\cap M\vdash_\KBeq B\psi\to B\lnot\varphi$ by (KBM).  So since
    $\check B\varphi=\lnot B\lnot\varphi$, it follows by classical
    reasoning that
    $w\cap M\vdash_\KBeq \check B\varphi\to\lnot B\psi$.  Since
    $w\cap M\vdash_\KBeq\check B\varphi$, it follows that
    $w\cap M\vdash_\KBeq\lnot B\psi$.  By the consistency of $\KBeq$
    (which follows by applying soundness to any mid-threshold
    epistemic neighborhood model), we have $w\cap M\nvdash_\KB B\psi$.
    So we have shown that $w\cap M\nvdash_\KBeq B\psi$ for each
    $\psi\in C$ satisfying $X'=[\psi]\cap[w]$.  Conclusion:
    $X'\notin N(w)$.

  \item[(sc)] If $X'\notin N(w)$ and $X\subsetneq Y\subseteq[w]$,
    then $Y\in N(w)$.

    Assume $X'\notin N(w)$ and $X\subsetneq Y\subseteq[w]$.  It
    follows from $X'\notin N(w)$ that we have
    $w\cap M\nvdash_\KBeq B(X')^d$.  Since $(X')^d\in C_0$, it follows
    that $B(X')^d\in C=\MCl(C_0)$ and therefore
    $\lnot B(X')^d\in w\cap M$ because $w$ is maxcons in $C$.  Hence
    $w\cap M\vdash_\KBeq \lnot B(X')^d$.  Now we have
    $w\cap M\vdash_\KBeq \lnot X^d\to (X')^d$, for otherwise we could
    extend $(w\cap M)\cup\{\lnot X^d,\lnot(X')^d\}$ to some $v\in[w]$
    such that $v\in[\lnot X^d]\cap[w]=X'$ and
    $v\in[\lnot(X')^d]\cap[w]=X$, contradicting $X'\cap X=\emptyset$.
    By (MP), we have $w\cap M\vdash_\KBeq K(\lnot X^d\to (X')^d)$ and
    therefore $w\cap M\vdash_\KBeq B\lnot X^d\to B(X')^d$ by (KBM).
    Since $w\cap M\vdash_\KBeq \lnot B(X')^d$, it follows by classical
    reasoning that $w\cap M\vdash_\KBeq\lnot B\lnot X^d$.  That is,
    $w\cap M\vdash_\KBeq\check B X^d$.

    Further, since $X\subsetneq Y\subseteq[w]$, it follows that there
    exists $y\in Y-X$ satisfying
    $y\in[Y^d]-[X^d]=[Y^d\land\lnot X^d]$.  Since
    $\lnot(Y^d\land\lnot X^d)\in C_0$, it follows that
    $\lnot K\lnot(Y^d\land\lnot X^d)\in C=\MCl(C_0)$.  But then
    $K\lnot(Y^d\land\lnot X^d)\notin w$, for otherwise it would follow
    from $y\in[w]$ that $w\cap M=y\cap M$ and hence
    $K\lnot(Y^d\land\lnot X^d)\in y$, from which it would follow by
    $\mathsf{T}$ and the fact that $y$ is maxcons in $C$ that
    $\lnot(Y^d\land\lnot X^d)\in y$, contradicting
    $y\in [Y^d\land\lnot X^d]$.  So since
    $K\lnot(Y^d\land\lnot X^d)\notin w$, we have by the fact that
    $\lnot K\lnot(Y^d\land\lnot X^d)\in C$ and the maximal
    $\KBeq$-consistency of $w$ that
    $\lnot K\lnot (Y^d\land\lnot X^d)= \check K(Y^d\land\lnot X^d)\in
    w$.
    Hence $w\cap M\vdash_\KBeq \check K(Y^d\land\lnot X^d)$. As
    $w\cap M\vdash_\KBeq\check B X^d$ as well, it follows by (SC) that
    $w\cap M\vdash_\KBeq B(Y^d\lor X^d)$.  But $[Y^d\lor X^d]=Y$ by
    our assumption $X\subsetneq Y$ and therefore we have shown that
    $Y\in N(w)$.

  \item[(scott)] If $X_1,\dots,X_m,Y_1,\dots,Y_m\subseteq[w]$,
    $(X_i\mathbb{I}Y_i)_{i=1}^m$, $X_1\in N(w)$, and
    $X_i':=[w]-X_i\notin N(w)$ for all $i\in\{2,\dots,m\}$, then there
    exists $j\in\{1,\dots,m\}$ such that $Y_j\in N(w)$.

    Assume we have the above-stated antecedent of the (scott)
    property.  It follows from $X_1\in N(w)$ that
    $w\cap M\vdash_\KBeq BX_1^d$.  For $i\in\{2,\dots,m\}$, it follows
    from $X_i'\notin N(w)$ by an argument as in the above proof for
    (sc) that $w\cap M\vdash\check BX_i^d$. If we can prove that
    $w\cap M\vdash_\KBeq (X_i^d\mathbb{I}Y_i^d)_{i=1}^m$ as well, then
    we would have by (Scott) that
    $w\cap M\vdash_\KBeq\bigvee_{j=1}^m BY_j^d$. But then since
    $BY_j^d\in C$ for each $j\in\{1,\dots,m\}$, we would have
    $BY_k^d\in w$ for some $k\in\{1,\dots,m\}$ by the maximal
    $\KBeq$-consistency of $w$, hence $w\cap M\vdash_\KBeq BY_k^d$,
    and hence $Y_k=[Y_k^d]\in N(w)$.

    So it suffices for us to prove that
    $w\cap M\vdash_\KBeq (X_i^d\mathbb{I}Y_i^d)_{i=1}^m$. Proceeding,
    we recall that $(X_i^d\mathbb{I}Y_i^d)_{i=1}^m$ abbreviates the
    formula $K(F_0\lor\cdots\lor F_m)$, where $F_k$ is the disjunction
    of all conjunctions
    \begin{equation}
      d_1X_1^d\land\cdots\land d_mX_m^d\land e_1Y_1^d\land\cdots\land
      e_mY_m^d\enspace,
      \label{eq1:KBeq-completeness}
    \end{equation}
    satisfying the property that exactly $k$ of the $d_i$'s are the
    empty string, at least $k$ of the $e_i$'s are the empty string,
    and the rest of the $d_i$'s and $e_i$'s are the negation sign
    $\lnot$. Since each of the $X_i$'s and $Y_i$'s is a member of
    $C_0$ and $C_0$ is closed under Boolean operations, each
    conjunction \eqref{eq1:KBeq-completeness} is a member of $C_0$,
    and hence so is the disjunction $F_0\lor\cdots\lor F_m$.  But then
    $K(F_0\lor\cdots\lor F_m)\in C=\MCl(C_0)$. We make use of these
    facts tacitly in what follows.  Now we have by our assumption
    $(X_i\mathbb{I}Y_i)_{i=1}^m$ and the fact that the $X_i$'s and
    $Y_i$'s are subsets of $[w]$ that every world in $[w]$ is
    contained in at least as many of the $X_i$'s as in the $Y_i$'s.
    Every world in $[w]$ therefore contains at least one of the
    $F_i$'s, for otherwise it would follow by maximal
    $\KBeq$-consistency that we could find a world $v\in[w]$ that is
    not contained in at least as many of the $X_i$'s as in the
    $Y_i$'s, a contradiction.  By maximal $\KBeq$-consistency, every
    world in $[w]$ thereby contains the disjunction
    $F_0\lor\cdots\lor F_m$.  But then it follows by maximal
    $\KBeq$-consistency and $\mathsf{T}$-reasoning that
    $K(F_0\lor\cdots\lor F_m)\in w$, and hence
    $w\cap M\vdash_\KBeq (X_i^d\mathbb{I}Y_i^d)_{i=1}^m$.  \qedhere
  \end{description}
\end{proof}

Since $\KBeq$ is sound and complete with respect to mid-threshold
neighborhood models, we would expect from
Corollary~\ref{corollary:lenzen} that $\KBeq$ is sound and complete
with respect to the probability interpretation for threshold $c=\frac
12$.

\begin{theorem}[Due to \cite{Lenzen1980:gwuw}; $\KBeq$
  Probability Soundness and Completeness]
  $\KBeq$ is sound and complete for threshold $\frac12$ with respect
  to the class of epistemic probability models:
  \[
  \textstyle 
  \forall\varphi\in\Lang_\KB:\quad
  \KBeq\vdash\varphi
   \quad\Leftrightarrow\quad
  {}\modelsp\varphi^{\frac 12} \enspace.
  \]
\end{theorem}
\begin{proof}
  Soundness is by Theorems~\ref{theorem:knowledge} and
  \ref{theorem:belief}.  Completeness is by Theorem~\ref{theorem:KBeq}
  and Corollary~\ref{corollary:lenzen}.
\end{proof}

\section{Conclusion} 
\label{SectionFRW}

\paragraph{Summary}

We have provided a study of unary modal logics of high probability.
We introduced epistemic neighborhood models and studied their
connection to traditional epistemic probability models by way of a
natural notion of ``agreement.''  We listed the Lenzen-derivative
properties of epistemic neighborhood models that guarantee the
existence of an agreeing probability measure for threshold
$c=\frac 12$.  The list of properties required to guarantee the
existence of an agreeing probability measure for other thresholds is
unknown.  We also presented our study from a proof theoretic point of
view by introducing a probabilistically sound but incomplete logic
$\KB$ and our version of Lenzen's probabilistically sound and complete
logic $\KBeq$ for threshold $c=\frac 12$.  It is open as to the
principles one must add to $\KB$ in order to obtain probabilistic
completeness for other thresholds.  We also proved soundness and
completeness of $\KB$ and of $\KBeq$ with respect to a corresponding
class of epistemic neighborhood models.  The result for $\KBeq$ along
with our Theorem~\ref{theorem:lenzen}, a theorem we credit to Lenzen,
shows that $\KBeq$ is the logic of probabilistic certainty and of
probability exceeding $c=\frac 12$.  It is our hope that our
repackaging of Professor Lenzen's result will make his work more
accessible to a broad audience of modern modal logicians.  We also
hope that the connection we have made with neighborhood semantics will
prove useful in future work on modal logics of qualitative
probability.

\paragraph{Open Questions for Future Work}

\begin{enumerate}
\item The main open question is the following: given a
  ``high-threshold'' $c\in(\frac 12,1)\cap\Rat$, find the exact
  extension $\KB^c$ of $\KB$ that is probabilistically sound and
  complete for threshold $c$ with respect to the class of epistemic
  probability models, in the sense that we would have:
  \[
  \textstyle \forall\varphi\in\Lang_\KB:\quad \KB^c\vdash\varphi
  \quad\Leftrightarrow\quad \modelsp\varphi^c\enspace.
  \]
  Observing that (SC) and (Scott) are not valid for high-thresholds
  $c>\frac 12$, we conjecture that what is required are
  threshold-specific variants of (SC) and (Scott) that will together
  guarantee probability soundness and completeness.  Toward this end,
  we suggest the following schemes as a starting point:
\begin{center}
  \renewcommand{\arraystretch}{1.3}
  \begin{tabular}[t]{cl}
    (SC$_0^s$) &
    $\textstyle(\check K\varphi_0\land
    \bigwedge_{i=1}^s\check B\varphi_i\land
    \bigwedge_{i\neq j=0}^s K(\varphi_i\to\lnot \varphi_j))\to 
    B(\bigvee_{i=0}^s \varphi_i)$
    \\
    (SC$_1^s$) &
    $\textstyle(\bigwedge_{i=1}^s\check B\varphi_i\land
    \bigwedge_{i\neq j=1}^s K(\varphi_i\to\lnot \varphi_j))
    \to B(\bigvee_{i=1}^s \varphi_i)$
    \\
    (WS) &
    $\textstyle [(\varphi_i\mathbb{I}\psi_i)_{i=1}^m
    \land \bigwedge_{i=1}^m B\varphi_i] \to
    \bigvee_{i=1}^m B\psi_i$
  \end{tabular}
\end{center}
Observe that (SC) is just (SC$_0^1$).  Further, if we define
$s':=c/(1-c)$ and $s:=\text{ceiling}(s')$, then scheme (SC$_0^s$) is
probabilistically sound if $s=s'$ and scheme (SC$_1^s$) is
probabilistically sound if $s\neq s'$.  The reasoning for this is as
follows: $s'$ tells us the number of $(1-c)$'s that divide $c$.  In
particular, recall from Lemma~\ref{lemma:dual} that the probabilistic
interpretation of $\check B\varphi$ is that $\varphi$ is assigned
probability at least $1-c$.  Therefore, if we have $s$ disjoint
propositions that each have probability at least $1-c$, then the
probability of their disjunction will have probability
$s\cdot(1-c)\geq c$.  This inequality is strict if $s\neq s'$ and is
in fact an equality if $s=s'$.  Therefore, in the case $s\neq s'$,
scheme (SC$_1^s$) is sound: $s$ disjoint propositions each having
probability $1-c$ together sum to a probability exceeding the
threshold $c$.  And in case $s=s'$, scheme (SC$_0^s$) is sound: $s$
disjoint propositions each having probability $1-c$ together sum to a
probability that equals $c$, so adding some additional probability
from another disjoint proposition $\varphi_0$ will yield a disjunction
whose probability again exceeds $c$.  In either case, exceeding
probability $c$ is what we equate with belief, so soundness is proved.
We note that scheme (WS) can be shown to be sound by adapting the
proof Theorem~\ref{theorem:belief}\eqref{item:B-Len}.  The epistemic
neighborhood model versions of (SC$_0^s$), (SC$_1^s$), and (WS) are:
\begin{description}
\item[(sc$_0^s$)] $\forall X_1,\dots,X_s,Y\subseteq[w]$:
  if $[w]-X_1,\dots,[w]-X_s\notin N(w)$,
  the $X_i$'s are pairwise disjoint, and $Y\supsetneq\bigcup_{i=1}^sX_i$,
  then $Y\in N(w)$.

\item[(sc$_1^s$)] $\forall X_1,\dots,X_s\subseteq[w]$: if
  $[w]-X_1,\dots,[w]-X_s\notin N(w)$ and the $X_i$'s are
  pairwise disjoint, then $\bigcup_{i=1}^sX_i\in N(w)$.
 
\item[(ws)] $\forall m\in\Int^+,\forall
  X_1,\dots,X_m,Y_1,\dots,Y_m\subseteq[w]:$
  \[
  \renewcommand{\arraystretch}{1.3}
  \begin{array}{ll}
    \text{if }
    &
    \begin{array}[t]{l}
      X_1,\dots,X_m\mathbb{I}Y_1,\dots,Y_m\quad\text{and}
      \\
      \forall i\in\{1,\dots,m\}:
      X_i\in N(w) \enspace\text{,}
    \end{array}
    \\
    \text{then }
    &
    \exists j\in\{1,\dots,m\}: Y_j\in N(w)\enspace\text{.}
  \end{array}
  \]
\end{description}
If $M$ is an epistemic neighborhood model, then a slight modification
of the proof of property (scott) in Lemma~\ref{lemma:correctness}
shows that $M^c$ satisfies (ws).  We presume that an adaptation of the
proof for the proof of property (sc) in the same lemma will show that
$M^c$ satisfies (sc$_0^s$) if $s=s'$ and (sc$_1^s$) if $s\neq s'$.

We remark that (WS) is not threshold-specific, though it is sound for
all high-thresholds $c>\frac 12$.  We suspect that a
threshold-specific variant may be required in order to adapt Lenzen's
proof of $\KBeq$ probability soundness and completeness for threshold
$c=\frac 12$ (Theorem~\ref{theorem:lenzen}).

\item Another open question is the exact relationship between
  Segerberg's comparitive operator $\varphi\preceq\psi$ (``$\varphi$
  is no more probable than $\psi$'')
  \cite{Gardenfors75,Segerberg1971:qpiams} and our unary operators $K$ and
  $B$.  The formula $B\varphi$ is equivalent to
  $\lnot\varphi\prec\varphi$.  However, it is not clear how the logics
  of these operators are related. Also, we suspect that a language
  with $\preceq$ is strictly more expressive.

\item Yet another direction is the extension of our work to Bayesian
  updating.  Given a pointed epistemic probability model $(\M,w)$
  satisfying $\varphi$, let
  \[
  \M[\varphi]=(W[\varphi],R[\varphi],V[\varphi],P[\varphi])
  \]
  be defined by
\begin{eqnarray*}
  W[\varphi] &:=& \semp{\varphi}^\M \\
  R[\varphi] &:=& R \cap (W[\varphi]\times W[\varphi]) \\
  V[\varphi](w) &:=& V(w) \text{ for } w\in W[\varphi] \\
  P[\varphi](w) &:=&
    \frac{ P(w) }{ P(\semp{\varphi}^\M) }
\end{eqnarray*}
It is not difficult to see that $\M[\varphi]$ is an epistemic probability model and
\[
P[\varphi](X) =
\frac{ P(X\cap\semp{\varphi}^\M) }{ P(\semp{\varphi}^\M) } = 
P[\varphi](X|\semp{\varphi}^\M)\enspace,
\]
where the value on the right is the probability of $X$ conditional on
$\semp{\varphi}^\M$.  It would be interesting to investigate the analog
of this operation in epistemic neighborhood models.  The operation may
also have a close relationship with the study of updates in
Probabilistic Dynamic Epistemic Logic
\cite{BenGerKoo09:SL,BalSme08:Synt}.

\item Finally, we have only considered a single-agent version of our
  logics $\KB$ and $\KBeq$.  The reason for this is that obtaining
  completeness for $\KBeq$ with respect to the class of finite
  mid-threshold neighborhood models requires us to construct a finite
  countermodel satisfying (sc), as we did in the completeness portion
  of the proof of Theorem~\ref{theorem:KBeq}.  However, this property
  has an antecedent that includes the negative condition
  $X'\notin N(w)$ and from this and $X\subsetneq Y\subseteq[w]$, we
  are to conclude the positive condition $Y\in N(w)$.  Referring the
  reader to the completeness portions of the proofs of
  Theorems~\ref{theorem:KB-neighborhood-completeness} and
  \ref{theorem:KBeq} for definitions and terminology, the trick to
  making things work in the single-agent case is to prove the
  Definability Lemma using a particular closure construction that
  ensures every potential neighborhood $X\subseteq[w]$ is definable by
  a formula $X^d$ such that $BX^d$ is a member of the closure set $C$.
  This makes crucial use of the In-class Identity Lemma.  However, our
  proof of this lemma depends on the assumption that maxcons sets
  $u,v\in[w]$ differ only in non-modal formulas.  In the
  straightforward multi-agent version of our setting, we would have an
  equivalence class $[w]_a$ consisting of all maxcons sets that agree
  on modal formulas $K_a\varphi$ and $B_a\psi$ for a given agent $a$.
  But then two worlds $u,v\in[w]_a$ could disagree on modal formulas
  $K_b\varphi$ or $B_b\psi$ for some agent $b\neq a$, which leads to a
  breakdown in the current proof of the In-class Identity Lemma and
  therefore presents problems for guaranteeing definability of
  potential neighborhoods satisfying the desired membership property.
  Remedying this in a multi-agent version of a finite mid-threshold
  neighborhood model is not straightforward because it is difficult to
  simultaneously satisfy (sc), all other properties of finite
  mid-threshold neighborhood models, and the
  definability-with-membership property. We therefore leave for future
  work the matter of proving completeness of multi-agent $\KBeq$ with
  respect to the class of finite multi-agent mid-threshold
  neighborhood models.  We note that multi-agent $\KBeq$ is obtained
  from our existing axiomatization by simply adding a subscript to all
  occurrences of a modal operator $K$ or $B$ in our present
  axiomatization. Multi-agent $\KB$ is obtained similarly, though
  completeness for multi-agent $\KB$ with respect to the full class of
  finite multi-agent epistemic neighborhood models can be shown
  without much difficulty because the problematic property (sc) need
  not be satisfied.
\end{enumerate}

\paragraph{Acknowledgements} 

Thanks to Alexandru Baltag, Johan van Benthem, Jim Delgrande, Peter
van Emde Boas, Andreas Herzig, Thomas Icard, Sonja Smets, and Rineke
Verbrugge for helpful comments and pointers to the literature.

\bibliographystyle{alpha}
\bibliography{vER-BEPL}  

\end{document}